\documentclass[11pt]{article}
\usepackage{amsmath,amssymb,amsthm,times}
\usepackage{fullpage}
\usepackage{enumerate}
\usepackage{color}

\newcommand{\calP}{\mathcal{P}}

\renewcommand{\d}{\mathrm{d}}

\newcommand{\eps}{\epsilon}
\newcommand{\ignore}[1]{}

\newcommand{\poly}{\mathrm{poly}}
\newcommand{\R}{\mathbb{R}}
\newcommand{\sgn}{\mathrm{sgn}}
\newcommand{\distance}{dist}
\newcommand{\myparagraph}[1]{\medskip\noindent {\bf #1\ }}
\newtheorem{theorem}{Theorem}[section]
\newtheorem{corollary}[theorem]{Corollary}
\newtheorem{lemma}[theorem]{Lemma}
\newtheorem{proposition}[theorem]{Proposition}

\theoremstyle{definition}
\newtheorem{definition}[theorem]{Definition}

\theoremstyle{plain}

\newcommand{\X}{\mathcal X} 
\renewcommand{\P}{\mathbb P} 
\newcommand{\prior}{\pi} 
\newcommand{\E}{\mathbb E} 
\newcommand{\nats}{\mathbb{N}} 
\newcommand{\reals}{\mathbb{R}} 

\newcommand{\citet}{\cite}
\newcommand{\citep}{\cite}

\begin{document}

\title{Active Property Testing}
\author{
Maria-Florina Balcan\footnote{
Georgia Institute of Technology, 
School of Computer Science.
Email: {\tt ninamf@cc.gatech.edu}. Supported in part by NSF grant CCF-0953192,
AFOSR grant FA9550-09-1-0538, a Microsoft Faculty Fellowship and a Google 
Research Award.} \and 
Eric Blais\footnote{
Carnegie Mellon University,
Computer Science Department.
Email: {\tt eblais@cs.cmu.edu}.} \and 
Avrim Blum\footnote{
Carnegie Mellon University,
Computer Science Department,
Email: {\tt avrim@cs.cmu.edu}.  Supported in part by the National Science 
Foundation under grants CCF-0830540, CCF-1116892, and IIS-1065251.} \and 
Liu Yang\footnote{
Carnegie Mellon University,
Machine Learning Department,
Email: {\tt liuy@cs.cmu.edu}. Supported in part by NSF grant IIS-1065251 and a 
Google Core AI grant.}}

\date{}
\maketitle

\begin{abstract}
One of the motivations for property testing of boolean functions is the idea
that testing can provide a fast preprocessing step before
learning.  However, in most machine learning applications, it is not possible
to request for labels of fictitious examples constructed by the algorithm.
Instead, the dominant query paradigm
in applied machine learning, called {\em active learning}, is one
where the algorithm
may query for labels, but {\em only on points in a given polynomial-sized 
(unlabeled) sample}, drawn from some underlying distribution $D$.
In this work, we bring this well-studied model in learning to the domain of
{\em testing}.  

We develop both general results for this {\em active testing} model as
well as efficient testing algorithms for a number of important
properties for learning, demonstrating that testing can still yield
substantial benefits in this restricted setting. For example, we show
that testing unions of $d$ intervals can be done with $O(1)$ label
requests in our setting, whereas it is known to require $\Omega(d)$
labeled examples for learning (and $\Omega(\sqrt{d})$ for passive
testing \cite{KR00} where the algorithm must pay for {\em every}
example drawn from $D$).  In fact, our results for testing unions of
intervals also yield improvements on prior work in both the classic
query model (where any point in the domain can be queried) and the
passive testing model as well. For the problem of testing linear
separators in $R^n$ over the Gaussian distribution, we show that both
active and passive testing can be done with $O(\sqrt{n})$ queries,
substantially less than the $\Omega(n)$ needed for learning, with
near-matching lower bounds.  We also present a general combination
result in this model for building testable properties out of
others, which we then use to provide testers for a number of
assumptions used in semi-supervised learning.

In addition to the above results, we also develop a general
notion of the {\em testing dimension} of a given property with respect
to a given distribution, that we show characterizes (up to
constant factors) the intrinsic number of label requests needed to
test that property.  We develop such notions for both the active 
and passive testing models.  We then use these dimensions to prove a 
number of lower bounds, including for linear separators and the class
of dictator functions.

Our results show that testing can be a powerful tool in realistic
models for learning, and further that active testing exhibits an
interesting and rich structure.  Our work in addition 
brings together tools 
from a range of areas including 
U-statistics,
noise-sensitivity,
self-correction, and
spectral analysis of random matrices,
and develops new tools that may be of
independent interest.

\end{abstract}

\setcounter{page}{0}
\thispagestyle{empty}
\newpage


\section{Introduction}

Property testing and machine learning have many natural
connections.  In property testing, given
black-box access to an unknown boolean function $f$, one would like 
with few queries to distinguish the case that $f$ has
some given property $\calP$ (belongs to the class of functions
$\calP$) from the case that $f$ is far from any function
having that property.  In
machine learning one would like to find a 
good approximation $g$ of $f$, typically under the assumption that $f$
belongs to a given class $\calP$.  This connection is in fact a natural
motivation for property testing: to cheaply determine
whether learning with a given hypothesis class is
worthwhile~\cite{GGR98,Ron08}.  If the labeling of examples is
expensive, or if a learning algorithm is computationally expensive to
run, or if one is deciding from what source to purchase one's data,
performing a cheap test in advance could be a substantial savings. 
Indeed, query-efficient testers have been designed for many
common function classes considered in machine learning
including 
linear threshold functions~\cite{MORS09a}, 
juntas~\cite{FKR+04,Bla09}, DNF
formulas~\cite{DLM+07}, and decision trees~\cite{DLM+07}.
(See Ron's survey~\cite{Ron08} for much more on
the connection between learning and property testing.)

However, there is a disconnect between the most commonly used
property-testing and machine learning models.  Most property-testing
algorithms rely on the ability to query functions on arbitrary points
of their choosing.  On the other hand, most machine learning problems
unfortunately do not allow one to perform queries on fictitious
examples constructed by an algorithm.  Consider, for instance, a
typical problem such as machine learning for medical diagnosis.  Given
a large database of patients with each patient described by various
features (height, 
age, family history, smoker or not, etc.), one would like to learn a
function that predicts from these features whether or not a patient
has a given medical condition (diabetes, for example). To perform this
learning task, the researchers can run a (typically expensive) medical test
on any of the patients to determine if the patient has the medical
condition.  However, researchers {\em cannot} ask whether the patient would
still have the disease were the values of some of his features
changed!  Moreover, researchers cannot make up a feature vector out of
whole cloth and 
ask if that feature vector has the disease.  As another example, in
classifying documents by topic, selecting an existing document on the
web and asking a labeler ``Is this about sports or business?'' may be
perfectly reasonable. However, the typical representation of a
document in a machine learning sytem is as a vector of word-counts in
$R^n$ (a ``bag of words'', without any information about the order in
which they appear in the document).  Thus, modifying some existing
vector, or creating a new one from scratch, would not produce an object
that we could expect a human labeler to easily classify.  The key
issue is that for most problems in machine learning, the example
and the label are in fact {\em both} functions of some underlying more
complex object.  Even in cases such as image classification---e.g.,
classifying handwritten digits into the numerals they
represent---where a human labeler would be examining the same
representation as the computer, queries can be problematic because the
space of reasonable images is a very sparse subset of the entire
domain.  Indeed, now-classic experiments on membership-query learning
algorithms for digit recognition ran into exactly this problem,
leading to poor results \cite{baum}.  In this case, the problem is
that the distribution one cares about (the distribution of natural
handwritten digits) is not one that the algorithm can easily 
construct new examples from.

As a result of these issues, the dominant query paradigm in 
machine learning in recent years is not one where the algorithm can
make arbitrary 
queries, but instead is a weaker model known as {\em active
learning}~\cite{SOS92,CAL,TonKol01,sanjoy,BBL,BDL09,CN07,Hanneke07,dhs,Kol10}.
In active learning, there is an underlying distribution $D$ over
unlabeled examples (say the distribution of documents on the web,
represented as vectors over word-counts) that we assume can be sampled
from cheaply: we assume the algorithm may obtain a polynomial number
of samples from $D$.  Then, the algorithm may ask an oracle for labels
(these oracle calls are viewed as expensive), {\em but only on points
in its sample}.  The goal of the active learning algorithm is to
produce an accurate hypothesis while requesting as few labels as
possible, ideally substantially fewer than in passive learning where
{\em every} example drawn from $D$ is labeled by the
oracle.  

In this work, we bridge this gap between testing and learning by
introducing, analyzing, and developing efficient algorithms for a
model of testing that parallels active learning, which we call {\em
active testing}.  As in active learning, we assume that our algorithm
is given a polynomial number of unlabeled examples from the underlying
distribution $D$ and can then make label queries, but {\em only} over
the points in its sample.  From a small number of such queries, the
algorithm must then answer whether the function has the given
property, or is far, with respect to $D$, from any function having
that property (see Section \ref{sec:model} for formal definitions).
We show that even with this restriction, we can still 
efficiently test important properties for machine learning including
unions of intervals, linear separators, and a number of properties
considered in semi-supervised learning.  Moreover, these testers
reveal important structural characteristics of these classes.  We
additionally develop a notion of {\em testing dimension} that 
characterizes the number of examples needed to test a given property
with respect to a given distribution, much like notions of dimension
in machine learning.  We do this for both the active testing model and
the weaker {\em passive testing} model 
\cite{GGR98,KR00} in which only random sampling of a small number
points from the distribution is allowed.  In fact, as part of our
analysis, we also develop improved algorithms for several important
classes for the passive testing model as well.  Overall, our results
demonstrate that active testing exhibits an interesting and rich
structure and strengthens the connection between testing and learning.

\subsection{Our Results}

We show that for a number of important properties for
learning---including unions of intervals, linear threshold functions,
and various assumptions used in semi-supervised learning---one can
test in the active testing model with substantially fewer
labels than needed to learn.  We in addition consider the even more
stringent {\em passive testing} model introduced by Goldreich,
Goldwasser, and Ron~\cite{GGR98} (where the only operation available
to the algorithm is to draw a random labeled sample from $D$) and give
new positive results for that model as well.  We further show that for
both active and passive testing models, we can {\em characterize} (up to
constant factors) the intrinsic number of label requests needed to
test any given property $\calP$ with respect to any given distribution
$D$ in a new quantity we call the {\em testing dimension} of $\calP$
with respect to $D$.
We then use these dimension notions to prove several near-tight lower
bounds.   We expand on each of these points below.

\myparagraph{Unions of intervals.}
The function $f : [0,1] \to \{0,1\}$ is a \emph{union of $d$
intervals} if the set $f^{-1}(1)$ consists of at most $d$ intervals in
$[0,1]$.  It is known that $\Theta(d)$ queries are
necessary and sufficient for {\em learning} functions from this class.
Kearns and Ron~\cite{KR00} showed that under the uniform
distribution, the relaxed problem of distinguishing
unions of $d$ intervals from functions that are $\epsilon$-far from unions of
$d/\eps$ intervals can be done with a constant number of queries in the standard
arbitrary-query testing model, and with $O(\sqrt{d})$ samples in the
passive testing model.  However, prior to the current work, no
non-trivial upper bound was known for the problem of distinguishing
unions of $d$ intervals from functions $\epsilon$-far from unions of
$d$ intervals (as opposed to far from $d/\epsilon$ intervals).

We give an algorithm that tests unions of $d$ intervals with only
$O(1)$ queries in the active testing model.  This result holds over
{\em any} underlying distribution (known or unknown).  Moreover, in
the case that the underlying distribution is uniform, we require only
$O(\sqrt{d})$ {\em unlabeled} samples.  Thus, as a byproduct we
improve over the prior best result in the passive testing model as
well.  Note that Kearns and Ron \cite{KR00} show that
$\Omega(\sqrt{d})$ examples are required to test unions of intervals
over the uniform distribution in the passive testing model, so this
result is tight.  Moreover, one can show that in the distribution-free
testing model of Halevy and Kushilevitz~\cite{HK07} one cannot perform
testing of this class from $O(1)$ queries; thus, this class
demonstrates a separation between these models (see Appendix
\ref{app:comparison}).

At the heart of the analysis of our algorithm is a characterization of
functions that are unions of intervals in terms of their \emph{noise
sensitivity}, shown via developing a local self-corrector for this
class.  The noise sensitivity of boolean functions
is a powerful tool that has led to recent
advances in hardness of approximation~\cite{KKMO07,MOO10}, learning theory~\cite{KOS04,KOS08,DHK+10},
and differential privacy~\cite{CKKL12}.  (See also~\cite{OD03} for more details on the
applications of noise sensitivity to the study of boolean
functions.)
Our work presents a novel application of noise sensitivity in the domain of property testing.

\myparagraph{Linear threshold functions.}
The function $f : \R^n \to \{0,1\}$ is a \emph{linear threshold
function} if there are $n+1$ parameters $w_1,\ldots,w_n, \theta \in
\R$ such that $f(x) = \sgn(w_1 x_1 + \cdots + w_n x_n - \theta)$ for
every $x \in \R^n$.  Linear threshold functions are perhaps the most
widely-used function class in machine learning.  
We show that both active and passive
testing of testing linear threshold functions in $R^n$ can be done
with $O(\sqrt{n})$ labeled examples over
the Gaussian distribution.
This is substantially less than the $\Omega(n)$ labeled examples
needed for learning (even over the Gaussian distribution
\cite{Long95}) and yields a new upper bound for the passive testing
model as well.  The key challenge here is that estimating a statistic due
to Matulef et al.~\cite{MORS09a}---which can be done with $O(1)$
queries if arbitrary queries are allowed \cite{MORS09a}---would
require $\Theta(n)$ samples if done from independent pairs of random
examples in the natural way; this is no better than learning.  We overcome this
obstacle by re-using non-independent pairs of examples in the
estimation, together with an analysis and modification of the
statistic that allow for use of a theorem of Arcones \cite{Arc95} on
the concentration of U-statistics.  At a technical level, this result
uses the fact that even though typical values of $(x \cdot y)^2$ may
be quite large---i.e., $\Theta(n)$---when $x$ and $y$ have every
coordinate selected from the standard normal, for any boolean function
$f$ it will be the case that for ``most'' values $y$, the quantity
$(\E_x[f(x)x \cdot y])^2$ is quite small---which can be shown via a
Fourier decomposition of $f$.  This in turn allows one to show strong
concentration.

Interestingly, we show these bounds are nearly tight, giving lower
bounds of $\tilde{\Omega}(n^{1/3})$ and $\tilde{\Omega}(\sqrt{n})$ on
the number of labeled examples needed for active and passive testing
respectively.  The proof of these lower bounds relies on our notion of
active and passive {\em testing dimensions}.  More precisely, by using the
notion of dimension, we reduce the problem of proving the lower
bounds to that of bounding the operator norm of random matrices.  This task
is then completed by appealing to recent results on the non-asymptotic
analysis of random matrices~\cite{Ver11}.
Our lower bound demonstrates a separation between
the active model and the standard (arbitrary-query) testing model.

\myparagraph{Disjoint unions of testable properties.}
We also show that any disjoint union of testable properties remains
testable in the active testing model, allowing one to build testable
properties out of simpler components; this is then used to provide
label-efficient testers for several properties used in semi-supervised
learning including the cluster and margin assumptions.  See Section
\ref{sec:disjoint} for details. 



\myparagraph{Testing dimension.}
One of the most powerful notions in learning theory is that of the 
\emph{dimension} or \emph{intrinsic complexity} of a class of
functions.  Such notions of dimension (e.g., VC dimension \cite{VC}, SQ
dimension \cite{BFJKMR94}, Rademacher complexity \cite{BartlettM02})
have been exceedingly 
effective in
determining the sample complexity for learning classes of functions in
various learning models.  Y.~Mansour and G.~Kalai (personal
communication, see also \cite{Kalai03}) posed the question
of whether comparable notions of dimension might exist for testing.
In this work, we answer in the affirmative and introduce the first
such notions of dimension for property testing, for both our new model
of active testing and the passive testing model.

We show that these notions of testing dimension characterize (up to
constant factors) the intrinsic number of labeled examples required to
test the given property with respect to a given distribution in the active and
passive testing models, respectively.  We also introduce a simpler
``coarse'' notion of testing dimension that characterizes the set of
properties testable with $O(1)$ queries in the active
testing model.

We use these testing dimensions to obtain lower bounds on the query
complexity for testing a number of different properties in both active
and passive testing models.  Notably, we show that $\Omega(\log n)$
queries are needed to distinguish dictator functions from random
functions in both models.  This shows that active testing of
dictators is as hard as learning dictator functions, and also implies
a lower bound of $\Omega(\log n)$ queries for testing a large
number of properties---including decision trees, functions of low
Fourier degree, juntas, DNFs---in the active testing
model.\footnote{Building on this analysis, Noga Alon (personal
communication) has recently developed a stronger $\Omega(k \log n)$
lower bound for the active testing dimension of juntas via use of the Kim-Vu
polynomial method.}

\subsection{Related Work}
\label{sec:related}
{\bf Active learning.}
Active learning has become a topic of substantial importance in machine
learning due to the rise of applications in which unlabeled data can
be sampled much more cheaply than data can be labeled, including text
classification \cite{McNi98,TonKol01}, medical imaging \cite{IKMTC11},
and image and music retrieval \cite{TongC01,MandelPE06} among
many others \cite{Gao11,vondr,WiensG10}.  This has led to significant work in algorithmic development including a yearly active-learning competition, with monetary prizes.\footnote{See
http://www.causality.inf.ethz.ch/activelearning.php.}  There has also been
substantial progress in the theoretical understanding of its
underlying principles,
including both algorithmic guarantees and the design and analysis
of appropriate sample complexity measures for this setting
\cite{QBC,BBL,BBZ07,BDL09,CN07,greedy,dkm,dhs,Hanneke07,Kol10,BeygelzimerHLZ10,ABE11,Ailon11,Hanneke11,Minsker12}. 
Active learning, unlike passive learning, has no known strong
Structural Risk Minimization bounds, which further motivates our work.
We note that while our model is motivated by
active learning, our techniques are very different
from those in the active learning literature.
\medskip

\noindent
{\bf Other Testing Models.}
In addition to the standard model of property testing~\cite{RS96} and
the passive model of property testing~\cite{GGR98,KR00} discussed
above, other models have been introduced to address different testing
scenarios.  The \emph{tolerant} testing model, introduced by Parnas,
Ron, and Rubinfeld~\cite{PRR06} was introduced to model situations
where the tester must not only accept functions that have a given
property but also must accept functions that are close to
having the property.  The \emph{distribution-free} testing model was
introduced by Halevy and Kushilevitz~\cite{HK07} (see also
\cite{HK04,HK05, GS09,DR10}) to explore the setting where the tester
does not know the underlying distribution $D$.  Both of these models
allow arbitrary queries, however, and so do not address the machine learning settings motivating this work in which one can only query inputs from a large sample of unlabeled points.    In Appendix~\ref{app:comparison}, we discuss
the technical relations between active testing and these other models.

\section{The Active Property Testing Model}
\label{sec:model}

A \emph{property} $\calP$ of
 boolean functions is simply a subset of
all boolean functions.  We will also refer to properties as
\emph{classes} of functions.  The \emph{distance} of a function 
$f$ to the property $\calP$ with respect to a distribution $D$ over the
domain of the function is
$
\distance_D(f, \calP) := \min_{g \in \calP} \Pr_{x \sim D}[ f(x) \neq g(x) ].
$
A \emph{tester} for $\calP$ is a randomized algorithm that must
distinguish (with high probability) between functions in $\calP$ and
functions that are far from $\calP$. In the standard property testing
model introduced by Rubinfeld and Sudan~\cite{RS96}, a
tester is allowed to query the value of the function on any input
in order to make this decision.  We consider instead a 
model in which we add restrictions to the possible queries:

\begin{definition}[Property tester] 
\label{def:pt}
An \emph{$s$-sample, $q$-query $\eps$-tester}
for $\calP$ over the distribution $D$ is a randomized algorithm $A$
that draws a sample $S$ of size $s$ from $D$, 
queries for the value of
$f$
on $q$ points of $S$, and then
\begin{enumerate}
\vspace{-3pt}
\setlength{\itemsep}{1pt}
\setlength{\parskip}{0pt}
\item Accepts w.p. at least $\frac23$ when $f \in \calP$, and
\item Rejects w.p. at least $\frac23$ when $\distance_D(f,\calP) \ge \eps$.
\end{enumerate}
\end{definition}

We will use the
terms ``label request'' and ``query'' interchangeably.
Definition~\ref{def:pt} coincides with the standard definition of
property testing when the number of samples is unlimited and the 
distribution's support covers the entire domain. In the other extreme
case where we fix $q = s$, our definition then corresponds to the
\emph{passive testing} model of Goldreich, Goldwasser, and Ron \cite{GGR98}, 
where the inputs queried by the tester 
are sampled from the distribution.  Finally, by setting $s$ to be
polynomial in an appropriate measure of the input domain or property
$\calP$, we obtain the \emph{active testing} model that is the focus
of this paper: 

\begin{definition}[Active tester]
A randomized algorithm is a \emph{$q$-query active $\eps$-tester} for 
$\calP \subseteq \{0,1\}^n \to \{0,1\}$ over
$D$ if it is a $\poly(n)$-sample, $q$-query $\eps$-tester for $\calP$ over $D$.\footnote{
We emphasize that the name \emph{active tester} is chosen to reflect the connection
with active learning. It is \emph{not} meant to imply that this model of testing
is somehow ``more active'' than the standard property testing model.}
\end{definition}

In some cases, the domain of our functions is not $\{0,1\}^n$.  In those
cases, we require $s$ to be polynomial in some other appropriate 
measure of complexity of the domain or property $\calP$ that we
specify explicitly. 
Note that in Definition~\ref{def:pt}, since we do not have direct 
membership query access (at arbitrary points), 
our tester must accept w.p.~at least $\frac23$ when $f$ is such that
$\distance_D(f, \calP)=0$, even if $f$ does not satisfy $\calP$ over the entire
input space.    See Appendix~\ref{app:comparison} for a comparison of
active testing to other testing models.

\section{Testing Unions of Intervals}
\label{sec:ui}

\newcommand{\calI}{\mathcal{I}}
\newcommand{\dist}{\mathrm{d}}
\newcommand{\NS}{\mathbb{NS}}
\newcommand{\NSrm}{\mathrm{NS}}

\label{sec:ui-active}

The function $f : [0,1] \to \{0,1\}$ is a \emph{union of $d$
intervals} if there are at most $d$ non-overlapping intervals
$[\ell_1,u_1], \ldots, [\ell_d, u_d]$ such that $f(x) = 1$ iff $\ell_i
\le x \le u_i$ for some $i \in [d]$.  The VC dimension of this class
is $2d$, so learning a union of $d$ intervals requires 
$\Omega(d)$ queries.  By contrast, we show that active testing of unions of $d$
intervals can be done with a number of label requests that is
\emph{independent} of $d$, for any (even unknown) distribution $D$.
Specifically, 
we prove that we can test unions of $d$ intervals in the active
testing model using only $O(1/\epsilon^4)$ label requests from a set of $poly(d,1/\epsilon)$ unlabeled examples.  Furthermore, over the uniform distribution, we need a total of only 
$O(\sqrt{d}/\epsilon^5)$ unlabeled examples.  Note that previously it
was not known how to test this class from $O(1)$ queries even in the
(standard) membership query model even over the uniform
distribution.\footnote{The best prior result achieved a relaxed guarantee of
distinguishing the case that $f$ 
is a union of $d$ intervals from the case that $f$ is $\epsilon$-far
from a union of $d/\epsilon$ intervals \cite{KR00}.}

\begin{theorem}
\label{thm:ui-act}
For any (known or unknown) distribution $D$, testing unions of $d$ intervals in
the active testing model can be 
done using only $O(1/\epsilon^{4})$ queries.  In the case of the
uniform distribution, we further need only $O(\sqrt{d}/\epsilon^5)$
unlabeled examples.
\end{theorem}

We prove Theorem \ref{thm:ui-act} by beginning with the case that the
underlying distribution is uniform over $[0,1]$, and afterwards show
how to generalize to arbitrary distributions.
Our tester is based on showing that unions of intervals have a
\emph{noise sensitivity} characterization.

\begin{definition}
Fix $\delta > 0$. The \emph{local $\delta$-noise
sensitivity} of the function $f : [0,1] \to \{0,1\}$ at $x \in [0,1]$ is
$
\mathrm{NS}_\delta(f,x) = \Pr_{y \sim_\delta x}[ f(x) \neq f(y) ],
$
where $y \sim_\delta x$ represents a draw of $y$ uniform
in $(x-\delta,x+\delta) \cap [0,1]$.
The \emph{noise sensitivity} of $f$ is
$$
\NS_\delta(f) 
=  \Pr_{x, y \sim_\delta x}[ f(x) \neq f(y) ]
$$
or, equivalently, $\NS_\delta(f) = \E_x \NSrm_\delta(f,x)$.
\end{definition}

A simple argument shows that unions of $d$ intervals have (relatively)
low noise sensitivity:

\begin{proposition}
\label{prop:NSInt}
Fix $\delta > 0$ and let $f : [0,1] \to \{0,1\}$ be a union of $d$ intervals.  Then
 $\NS_\delta(f) \le d \delta$.
\end{proposition}

\begin{proof}[Proof sketch]
Draw $x \in [0,1]$ uniformly at random and $y \sim_\delta x$.  The
inequality $f(x) \neq f(y)$ can only hold when a boundary
$b \in [0,1]$ of one of the $d$ intervals in $f$ lies in between
$x$ and $y$.  For any point $b \in [0,1]$, the probability that
$x < b < y$ or $y < b < x$ is at most $\frac\delta2$, and there
are at most $2d$ boundaries of intervals in $f$, so the proposition
follows from the union bound.
\end{proof}

The key to the tester is showing that the converse of the above
statement is  approximately true as well:
for $\delta$ small enough, every function that has
noise sensitivity not much larger than $d \delta$ is close to being a
union of $d$ intervals. (Full proof in Appendix \ref{app:ui}).

\begin{lemma}
\label{lem:NSnonInt}
Fix $\delta = \frac{\eps^2}{32d}$.  Let $f : [0,1] \to \{0,1\}$ be a function with
noise sensitivity bounded by $\NS_\delta(f) \le d\delta(1+\frac{\eps}4)$.  Then $f$ is
$\eps$-close to a union of $d$ intervals.
\end{lemma}

\begin{proof}[Proof outline]
The proof proceeds in two steps.  First, we show that so long as $f$
has low noise-sensitivity, it can be ``locally self-corrected'' to 
a function
$g : [0,1] \to \{0,1\}$ that is $\frac\eps2$-close to $f$ and is a union of at most
$d(1 + \frac\eps4)$ intervals.  We then show that $g$ -- and every other
function that is a union of at most $d(1 + \frac\eps4)$ intervals -- is
$\frac\eps2$-close to a union of $d$ intervals.

To construct the function $g$, we consider a smoothed function
$f_\delta : [0,1] \to [0,1]$ obtained by taking the convolution of $f$ and a uniform
kernel of width $2\delta$.  We define $\tau$ to be some appropriately
small parameter.  When $f_\delta(x) \le \tau$, then this means that nearly all the
points in the $\delta$-neighborhood of $x$ have the value $0$ in $f$,
so we set $g(x) = 0$.  Similarly, when $f_\delta(x) \ge 1 - \tau$, then
we set $g(x) = 1$.  (This procedure removes any ``local noise'' that
might be present in $f$.)  This leaves all the points $x$ where $\tau <
f_\delta(x) < 1 - \tau$.  Let us call these points \emph{undefined}. For
each such point $x$ we take the largest value
$y \le x$ that is defined and set $g(x) = g(y)$.
The key technical part of the proof involves showing that the construction
described above yields a function $g$ that is $\frac\eps2$-close to $f$
and that is a union of $d(1 + \frac\eps4)$ intervals.  
Due to space constraints, we defer the argument to Appendix~\ref{app:ui}.
\end{proof}

The noise sensitivity characterization of unions of intervals obtained
by Proposition~\ref{prop:NSInt} and Lemma~\ref{lem:NSnonInt}
suggest a natural approach for building a tester: design an
algorithm that estimates the noise sensitivity of the input function
and accepts iff this noise sensitivity is small enough.  This is
indeed what we do:

\smallskip

\begin{quote}
{\sc Union of Intervals Tester}( $f$, $d$, $\eps$ ) \\
\qquad Parameters: $\delta = \frac{\eps^2}{32d}, r = O(\eps^{-4})$.
\begin{enumerate}
\vspace{-8pt}
\setlength{\itemsep}{1pt}
\item For rounds $i = 1,\ldots, r$,
\begin{enumerate}
\item[1.1] Draw $x \in [0,1]$ uniformly at random.
\item[1.2] Draw samples until we obtain $y \in (x-\delta,x+\delta)$.
\item[1.3] Set $Z_i = \mathbf{1}[ f(x) \neq f(y) ]$.
\end{enumerate}
\item {\bf Accept} iff $\frac1r \sum Z_i \le d \delta( 1 + \frac\eps8)$.
\end{enumerate}
\end{quote}

\smallskip
The algorithm makes $2r = O(\eps^{-4})$ queries to the function.
Since a draw in Step 1.2 is in the desired range with probability
$2\delta$, the number of samples drawn by the
algorithm is a random variable with very tight concentration
around $r(1 + \frac1{2\delta}) = O(d/\eps^6)$.  The draw in Step 1.2
also corresponds to choosing $y \sim_\delta x$.
  As a result, the probability that
$f(x) \neq f(y)$ in a given round is exactly $\NS_\delta(f)$,
and the average $\frac1r \sum Z_i$ is an unbiased
estimate of the noise sensitivity of $f$.  By Proposition~\ref{prop:NSInt},
Lemma~\ref{lem:NSnonInt}, and Chernoff bounds, the
algorithm therefore errs with probability less than $\frac13$
provided that $r > c \cdot 1/(d\delta\eps^2) = c \cdot 32/\eps^4$ for
some suitably large constant $c$.

\medskip
\noindent{\bf Improved unlabeled sample complexity:}
Notice that by changing Steps 1.1-1.2 slightly to pick
the first pair $(x,y)$ such that $|x-y| < \delta$, we immediately improve
the unlabeled sample complexity to $O(\sqrt{d}/\eps^5)$ without affecting
the analysis.  In particular, this procedure is equivalent to picking
$x \in [0,1]$  then $y \sim_\delta x$.\footnote{Except for events of
$O(\delta)$ probability mass at the boundary.}  As a result, up to
$poly(1/\epsilon)$ terms, we also
improve over the {\em passive testing} bounds of Kearns and Ron
\cite{KR00} which are able only to distinguish the case that $f$ is a
union of $d$ intervals from the case that $f$ is $\epsilon$-far from
being a union of $d/\epsilon$ intervals.  (Their results use
$O(\sqrt{d}/\epsilon^{1.5})$ examples.)  Kearns and Ron
\cite{KR00} show that $\Omega(\sqrt{d})$ examples are necessary for
passive testing, so in terms of $d$ this is optimal.

\medskip
\noindent{\bf Active testing over arbitrary distributions:} We now
consider the case that examples are drawn from some arbitrary
distribution $D$.  First, let us consider the easier case that $D$ is
known.  In that case, we can reduce the problem of testing over
general distributions to that of testing over the uniform distribution
on $[0,1]$ by using the CDF of $D$.  In particular, given
point $x$, define $p_x =
\Pr_{y\sim D}[y \leq x]$.  So, for $x$ drawn from $D$, $p_x$ is
uniform in $[0,1]$.\footnote{We are assuming here that $D$ is
continuous and has a pdf.  If $D$ has point masses, then instead
define $p_x^L = \Pr_y[y<x]$ and $p_x^U = \Pr_y[y \leq x]$ and select
$p_x$ uniformly in $[p_x^L,p_x^U]$.}  As a result we can just replace
Step 1.2 in the tester with 
sampling until we obtain $y$ such that $p_y \in (p_x - \delta, p_x +
\delta)$.  Now, suppose $D$ is not known.  In that case, we do not
know the $p_x$ and $p_y$ values exactly.  However, we can use the fact
that the VC-dimension of the class of initial intervals on the line
equals 1 to uniformly estimate all such values from a polynomial-sized
unlabeled
sample.  In particular, $O(1/\gamma^2)$ unlabeled examples are sufficient
so that with high probability, {\em every} point $x$ has property that
the estimate $\hat{p}_x$ of $p_x$ computed with respect to the
sample (the fraction of
points in the {\em sample} that are $\leq x$) will be within
$\gamma$ of the correct $p_x$ value \cite{BEHW:JACM89}.
If we define $\hat{\NS}_\delta(f)$ to be the noise-sensitivity
of $f$ computed using these estimates, then we get
$\frac{\delta-\gamma}{\delta+\gamma}\NS_{\delta-\gamma}(f) \leq
\hat{\NS}_\delta(f) \leq
\frac{\delta+\gamma}{\delta - \gamma}\NS_{\delta+\gamma}(f)$.  This implies
that $\gamma = O(\epsilon\delta)$ is sufficient so that 
the noise-sensitivity estimates are
sufficiently accurate for the procedure to work as before.

Putting these results together, we have Theorem~\ref{thm:ui-act}.


\section{Testing Linear Threshold Functions}
\label{sec:ltf}
\newcommand{\bbN}{\mathbb{N}}
\newcommand{\bX}{\mathbf{X}}
\newcommand{\bw}{\mathbf{w}}
\newcommand{\bz}{\mathbf{z}}
\newcommand{\calA}{\mathcal{A}}
\newcommand{\calN}{\mathcal{N}}
\newcommand{\Dno}{\mathcal{D}_{\mathrm{no}}}
\newcommand{\Dyes}{\mathcal{D}_{\mathrm{yes}}}
\newcommand{\high}{\mathrm{high}}
\newcommand{\low}{\mathrm{low}}
\newcommand{\op}{\mathrm{op}}
\newcommand{\vareps}{\varepsilon}


\ignore{
In the last section, we saw how unions of intervals are
characterized by a statistic of the function -- namely, its
noise sensitivity -- that can be estimated with few queries and
used this to build our tester.  In this section, we follow the same
high-level approach for testing linear threshold functions.  In this
case, however, the statistic we will estimate is not noise sensitivity
but rather the 
sum of squares of the degree-1 Hermite coefficients
of the function.
}

\newcommand{\comment}[1]{[{\it #1}]}

A boolean function $f : \R^n \to \{0,1\}$ is a \emph{linear threshold
function} (LTF) if there exist $n+1$ real-valued parameters
$w_1,\ldots,w_n,\theta$ such that for each $x \in \R^n$, we have $f(x)
= \sgn(w_1 x_1 + \cdots w_n x_n - \theta)$.\footnote{ Here, $\sgn(z) =
\mathbf{1}[ z \ge 0 ]$ is the standard sign function.} The main result
of this section is that it is possible to efficiently test whether a
function is a linear threshold function in the active and passive
testing models with substantially fewer labeled examples than needed
for learning, along with near-matching lower bounds.

\begin{theorem}
\label{thm:ltf}
We can efficiently test linear threshold functions under the Gaussian
distribution with $O(\sqrt{n \log n})$ labeled examples in both active 
and passive testing models.  Furthermore, no (even computationally
inefficient) algorithm can test with 
$\tilde{o}(n^{1/3})$ labeled examples for active testing or
$\tilde{o}(\sqrt{n})$ labeled examples for passive testing.
\end{theorem}

Note that the class of linear threshold functions requires
$\Omega(n)$ labeled examples for {\em learning}, even over the Gaussian
distribution~\cite{Long95}.  
Linear threshold functions can be tested with a
constant number of queries in the standard (arbitrary query) property
testing model \cite{MORS09a}.

The starting point for the upper bound in Theorem~\ref{thm:ltf} is 
a characterization lemma of linear threshold functions in terms of the following self-correlation
statistic.  To be precise, we are scaling so that 
each coordinate is drawn independently from $\calN(0,1)$---so a typical
example will have length $\Theta(\sqrt{n})$.

\begin{definition}
\label{def:selfcor}
The \emph{self-correlation coefficient} of the function $f : \R^n \to \R$ is
 $\rho(f) := \E_{x,y}[ f(x) f(y) \left<x,y\right>].$
\end{definition}

\begin{lemma}[Matulef et al.~\cite{MORS09a}]
\label{lem:MORS}
There is an explicit continuous function $W : \R \to \R$ with bounded
derivative $\|W'\|_\infty \le 1$ and peak value $W(0) = \frac2\pi$ such that every
linear threshold function $f : \R^n \to \{-1,1\}$ satisfies
$
\rho(f) = W( \E_x f ).
$
Moreover, every function $g : \R^n \to \{-1,1\}$ that satisfies
$
\left| \rho(g) - W( \E_x g) \right| \le 4 \eps^3,
$
is $\eps$-close to being a linear threshold function.
\end{lemma}

The proof of Lemma~\ref{lem:MORS} relies on the Hermite decomposition of 
functions. In fact, the original characterization of Matulef et al.~\cite{MORS09a} is stated
in terms of the level-1 Hermite weight of functions. The above characterization follows
easily from their result. For completeness, we include the details in Appendix~\ref{app:ltf}.

Lemma~\ref{lem:MORS} suggests an obvious approach to testing for
linear threshold functions from random examples: simply estimate the
self-correlation 
coefficient of Definition \ref{def:selfcor} by repeatedly drawing
pairs of labeled examples $(x_i,y_i)$ 
from the Gaussian distribution in $R^n$ and computing the empirical average of
the quantities $f(x_i) f(y_i) \left<x_i,y_i\right>$ observed.  The
problem with this approach, however, is that the dot-product
$\left<x_i,y_i\right>$ will typically have magnitude $\Theta(\sqrt{n})$
(one can view it as essentially the result of an $n$ step random
walk).  Therefore to estimate the self-correlation coefficient to
accuracy $O(1)$ via independent random samples in this way would
require $\Omega(n)$ labeled examples. 
This is of course not very useful, since it is the same as the number
of labeled examples needed to \emph{learn} an LTF. 

We will be able to achieve an improved bound, however, using the
following idea: rather than
averaging over independent pairs $(x,y)$, we will draw a smaller sample
and average over all (non-independent) pairs within the sample.  That
is, we request
$q$ random labeled examples $x_1,\ldots,x_q$, and now estimate 
$\rho(f)$ by computing ${q \choose 2}^{-1} \sum_{i < j} f(x_i) f(x_j)
\left< x_i, x_j\right>$.  
Of course, the terms in the summation are no longer independent.
However, they satisfy the property that even though the quantity
$f(x)f(y)\left<x,y\right>$ is typically large, for
most values $y$,  the quantity $\E_x[f(x)f(y)\left<x,y\right>]$ is
small.  (This can be shown via a Fourier decomposition of the function
$f$.)  
This, together with additional truncation of the quantity in question,
will allow us to apply a Bernstein-type inequality for U-statistics
due to Arcones~\cite{Arc95} in order to achieve the desired concentration.

The resulting {\sc LTF Tester} is given in Figure \ref{fig:ltf}.
This algorithm has two advantages. First, it is a valid tester in both the active and passive
property testing models since the $q$ inputs queried by the algorithm are all drawn independently
at random from the standard $n$-dimensional Gaussian distribution.  Second, the algorithm itself is
very simple.  As in many cases with property testing, however, the analysis of this algorithm is more 
challenging.

\begin{figure}[h!]
\begin{quote}
{\sc LTF Tester}( $f$, $\eps$ ) \\
\qquad Parameters: $\tau = \sqrt{4n\log(4n/\eps^3)}$, $m = 800\tau/\eps^3 + 32/\eps^6$.
\begin{enumerate}
\vspace{-8pt}
\setlength{\itemsep}{1pt}
\setlength{\parskip}{0pt}
\item Draw $x^{1},x^{2},\ldots,x^{m}$ independently at random from $\R^n$.
\item Query $f(x^{1}),f(x^{2}),\ldots,f(x^{m})$.
\item Set $\tilde{\mu} = \frac1m \sum_{i=1}^m f(x^{i})$.
\item Set $\tilde{\rho} = {m \choose 2}^{-1}
 \sum_{i \neq j} f(x^{i}) f(x^{j}) \left< x^{i}, x^{j} \right> \cdot \mathbf{1}[ \left|\left<x^i,x^j\right>\right| \le \tau]$.
\item {\bf Accept} iff $|\tilde{\rho} - W(\tilde{\mu})| \le 2\eps^3$.
\end{enumerate}
\vspace*{-0.3in}
\end{quote}\caption{\label{fig:ltf}{\sc LTF Tester}}
\end{figure}

Given Lemma~\ref{lem:MORS}, as noted above, the key challenge in
the proof of correctness of the {\sc LTF Tester} is 
controlling the error of the estimate $\tilde{\rho}$ of $\rho(f)$ in
Step 4, which we do  
with concentration of measure results for U-statistics.
The \emph{U-statistic} (of order 2) with symmetric kernel function $g : \R^n \times \R^n \to \R$ is 
$$
U_g^q(x^1,\ldots,x^q) := {q \choose 2}^{-1}\sum_{1 \le i < j \le q} g(x^i, x^j).
$$
U-statistics are unbiased estimators of the expectation of their kernel function and, even more importantly,
when the kernel function is ``well-behaved'', the tails of their
distributions satisfy strong concentration.
In our case, the thresholded kernel function $g(x,y) = \begin{cases}
f(x) f(y) \left< x, y \right> & |\left< x, y \right>| \le \tau \\ 0 &
\mbox{otherwise} \end{cases}$ \ \ \ \ allows us to apply Arcones' theorem.

\begin{lemma}[Arcones~\cite{Arc95}] 
\label{lem:arcones}
For a symmetric function $h : \R^n \times \R^n \to \R$, 
let $\Sigma^2 = \E_x[ \E_y[ h(x,y) ]^2 ] - \E_{x,y}[ h(x,y) ]^2$, let $b = \|h - \E h\|_{\infty}$, and
let $U_q(h)$ be a random variable obtained by drawing $x^1,\ldots,x^q$ independently at random
and setting $U_q(h) = {q \choose 2}^{-1} \sum_{i < j} h(x^{i},x^{j})$.
Then for every $t > 0$,
$$
\Pr[ |U_q(h) - \E h| > t ] \le 4 \exp\left( \frac{ qt^2 }{ 8 \Sigma^2 + 100bt } \right).
$$
\end{lemma}

An argument combining Lemma~\ref{lem:arcones} with a separate argument showing that $g$ is ``close'' to
an unbiased estimator for $\rho(f)$ provides the desired guarantee for
the {\sc LTF Tester}.  The complete proof is
presented in Appendix~\ref{app:ltf}.


\medskip
It is natural  to ask whether we can further improve the query complexity of the tester for
linear threshold functions by using U-statistics of higher order.  The lower bound in Theorem~\ref{thm:ltf}
shows that this---or any other possible active or passive testing
approach---cannot yield a query 
complexity sub-polynomial in $n$.  We defer the discussion of this lower bound to Section~\ref{sec:dim},
where we will use the notion of testing dimension to establish the bound.

\section{Testing Disjoint Unions of Testable Properties}
\label{sec:disjoint}
We now show that active testing has the feature that a disjoint union
of testable properties is testable, with a number of queries that is
independent of the size of the union; this feature does not hold for
passive testing.  In addition to providing insight into the
distinction between the two models, this fact will be useful in our
analysis of semi-supervised learning-based properties mentioned below
and discussed more fully in Appendix
\ref{sec:ssl}. 

Specifically, given properties $\calP_1, \ldots, \calP_N$ over domains
$X_1, \ldots, X_N$, define their disjoint union $\calP$ over domain $X
= \{(i,x): i
\in [N], x \in X_i\}$ to be the set of functions $f$ such that $f(i,x)
= f_i(x)$ for some $f_i \in \calP_i$.  In addition, for any distribution
$D$ over $X$, define $D_i$ to be the conditional distribution over
$X_i$ when the first component is $i$.  If each
$\calP_i$ is testable over $D_i$ then $\calP$ is testable over $D$
with only small overhead in the number of queries:
\begin{theorem}
\label{thm:disjoint}
Given properties $\calP_1, \ldots, \calP_N$,
if each $\calP_i$ is testable over $D_i$ with $q(\epsilon)$ queries
and $U(\epsilon)$ unlabeled samples, then their disjoint union $\calP$
is testable over the combined distribution $D$ with
$O(q(\epsilon/2)\cdot (\log^3 \frac{1}{\epsilon}))$ queries and
$O(U(\epsilon/2)\cdot (\frac{N}{\epsilon}\log^3\frac{1}{\epsilon}))$
unlabeled samples.
\end{theorem}
\begin{proof}
See Appendix \ref{app:disjoint}.
\end{proof}
As a simple example, consider $\calP_i$ to contain just the constant
functions {\bf 1} and {\bf 0}.  In this case, $\calP$ is equivalent to
what is often called the ``cluster assumption,'' used in
semi-supervised and active learning \cite{ssl:book06,Das11}, that if
data lies in some number of clearly identifiable clusters, then all
points in the same cluster should have the same label.  Here, each
$\calP_i$ individually is 
easily testable (even passively) with $O(1/\epsilon)$ labeled
samples, so Theorem \ref{thm:disjoint} implies the cluster assumption
is testable with $poly(1/\epsilon)$ queries.\footnote{Since
the $\calP_i$ are so simple in this case, one can actually test with
only $O(1/\epsilon)$ queries.}  However, it is not hard to see that
passive testing with $poly(1/\epsilon)$ samples is not possible and in
fact requires $\Omega(\sqrt{N}/\epsilon)$ labeled
examples.\footnote{Specifically, suppose region 1 has $1 - 2\epsilon$
probability mass with $f_1 \in \calP_1$, and suppose the other regions
equally share the remaining $2\epsilon$ probability mass and either
(a) are each pure but random (so $f \in \calP$) or (b) are each 50/50
(so $f$ is $\epsilon$-far from $\calP$).  Distinguishing these cases requires
seeing at least two points with the same index $i
\neq 1$, yielding the $\Omega(\sqrt{N}/\epsilon)$ bound.}

We build on this to produce testers for
other properties often used in semi-supervised learning.  
In particular, one common assumption used (often called the {\em margin}
or {\em low-density} assumption) is that  there should be some
large margin $\gamma$ of separation between the positive and negative
regions (but without 
assuming the target is necessarily a linear threshold function).
Here, we give a tester for this property, which uses 
a tester for the cluster property as a subroutine, along
with analysis of an appropriate weighted graph defined over the data.
Specifically, we prove the following result (See Appendix
\ref{sec:ssl} for definitions and analysis). 

\begin{theorem}\label{thm:margin}
For any $\gamma$, $\gamma' = \gamma(1-1/c)$ for constant $c>1$, for
data in the unit ball in $R^d$ for constant $d$, we can
distinguish the case that $D_f$ has margin $\gamma$ from the case that
$D_f$ is $\epsilon$-far from margin $\gamma'$ using Active Testing
with $O(1/(\gamma^{2d}\epsilon^{2}))$ unlabeled examples and
$O(1/\epsilon)$ label requests.
\end{theorem}



\section{General Testing Dimensions}
\label{sec:dim}

\newcommand{\Fair}{{\rm Fair}}
\newcommand{\dadaptive}{d_{aa}}
\newcommand{\dpassive}{d_{passive}}
\newcommand{\dcoarse}{d_{coarse}}
\newcommand{\dactive}{d_{active}}
\newcommand{\Priors}{\Pi}
\newcommand{\PinC}{\Priors_{0}}
\newcommand{\PfarC}{\Priors_{\epsilon}}
\newcommand{\pin}{\prior}
\newcommand{\pfar}{\prior^{\prime}}
\newcommand{\prand}{\prior^{{\rm rand}}}
\newcommand{\Good}{{\rm Good}}
\newcommand{\Bad}{{\rm Bad}}
\newcommand{\ind}{\mathbb{I}}

The previous sections have discussed upper and lower bounds for a
variety of classes.  Here, we define notions of
{\em testing dimension} for passive and active testing that
characterize (up to constant factors) the number of labels needed for
testing to succeed, in the corresponding testing protocols.
These will be distribution-specific notions (like SQ
dimension \cite{BFJKMR94} or Rademacher complexity \cite{BartlettM02}
in learning), so let us fix some distribution $D$ over the 
instance space $X$, and furthermore
fix some value $\epsilon$ defining our goal. I.e., our goal is to
distinguish the case that $\distance_D(f,\calP)=0$ from the case
$\distance_D(f,\calP) \geq \epsilon$.

For a given set $S$ of unlabeled points, 
and a distribution $\prior$ over boolean functions,
define $\prior_{S}$ to be the distribution over labelings of $S$ induced by $\prior$.
That is, for $y \in \{0,1\}^{|S|}$ let $\prior_S(y) = \Pr_{f \sim
\prior}[f(S)=y]$.
We now use this to define a distance between distributions.  Specifically,
given a set of unlabeled points $S$
and two distributions $\prior$ and $\prior^{\prime}$ over boolean functions,
define
\[\dist_{S}(\prior,\prior^{\prime}) = (1/2)\sum_{y\in \{0,1\}^{|S|}} | \prior_{S}(y) - \prior^{\prime}_{S}(y) |,\]
to be the variation distance between $\prior$ and $\prior'$ induced by
$S$.
Finally, let $\PinC$ be the set of all distributions
$\prior$ over functions in $\calP$, 
and let set $\PfarC$ be the set of all distributions $\prior'$ in
which a $1-o(1)$ probability mass is over functions at least
$\epsilon$-far from $\calP$.  
We are now ready to formulate our notions of dimension.

\subsection{Passive Testing Dimension}

\begin{definition}
Define the passive testing dimension, $\dpassive = \dpassive(\calP,D)$,
as the largest $q \in \nats$ such that, 
\[ \sup_{\pin \in \PinC} \sup_{\pfar \in \PfarC} \Pr_{S \sim D^q} ( 
\dist_S(\pin,\pfar) > 1/4 ) \leq 1/4.\] 
\end{definition}
That is, there exist distributions $\pin \in \PinC$ and $\pfar \in
\PfarC$ such that a random set $S$ of $\dpassive$ examples has a
reasonable probability (at least $3/4$) of having the property that
one cannot reliably distinguish a random function from $\pin$ versus a
random function from $\pfar$ from just the labels of $S$.  From the
definition it is fairly immediate that $\Omega(\dpassive)$ examples
are {\em necessary} for passive testing; in fact, one can show that
$O(\dpassive)$ are sufficient as well.
\begin{theorem}
\label{thm:passivedim}
The sample complexity of passive testing property $\calP$ over
distribution $D$ is $\Theta(\dpassive(\calP,D))$.
\end{theorem}
\begin{proof}
See Appendix \ref{app:dimension}.
\end{proof}
\paragraph{Connections to VC dimension.}
This notion of dimension brings out an interesting connection between
learning and testing.  In particular, consider the special case that
we simply wish to distinguish functions in $\calP$ from truly random
functions, so $\pfar$ is the uniform distribution over all functions
(this is indeed the form used by our lower bound results in Sections
\ref{sec:dict} and \ref{sec:ltf-dim}).  In that case, the passive
testing dimension becomes the largest $q$ such that for some
(multi)set $F$ of functions $f_i \in \calP$, a typical sample $S$ of
size $q$ would have all $2^q$ possible labelings occur {\em
approximately the same number of times} over the functions $f_i \in
F$.  In contrast, the {\em VC-dimension} of $\calP$ is the largest $q$
such that for some sample $S$ of size $q$, each of the $2^q$ possible
labelings occurs {\em at least once}.  Notice there is a kind of
reversal of quantifiers here: in a distributional version of
VC-dimension where one would like a ``typical'' set $S$ to be
shattered, the functions that induce the $2^q$ labelings could be
different from sample to sample.  However, for the testing dimension,
the set $F$ must be fixed in advance.  That is the reason that it is
possible for a tester to output ``no'' even though the labels observed
are still consistent with some function in $\calP$.

\subsection{Active Testing Dimension}

For the case of active testing, there are two complications. First,
the algorithms can examine their entire $poly(n)$-sized unlabeled
sample before deciding which points to query, and secondly they may in
principle determine the next query based on the responses to the
previous ones (even though all our algorithmic results do not require
this feature).  If we merely want to distinguish those properties that
are actively testable with $O(1)$ queries from those that are not,
then the second complication disappears and the first is simplified as
well, and the following coarse notion of dimension suffices.
\begin{definition}
Define the coarse active testing dimension, $\dcoarse = \dcoarse(\calP,D)$,
as the largest $q \in \nats$ such that, 
\[ \sup_{\pin \in \PinC} \sup_{\pfar \in \PfarC} \Pr_{S \sim D^q} ( 
\dist_S(\pin,\pfar) > 1/4 ) \leq 1/n^q.\] 
\end{definition}
\begin{theorem}
\label{thm:coarse}
If $\dcoarse(\calP,D) = O(1)$ the active testing of $\calP$ over $D$
can be done with $O(1)$ queries, and if $\dcoarse(\calP,D) = \omega(1)$ then it cannot.
\end{theorem}
\begin{proof}
See Appendix \ref{app:dimension}.
\end{proof}

To achieve a more fine-grained characterization of active testing we
consider a slightly more involved quantity, as follows.  First,
recall that given an unlabeled sample $U$ and distribution $\prior$
over functions, we define $\prior_U$ as the induced distribution over
labelings of $U$.  We can view this as a distribution over {\em unlabeled}
examples in $\{0,1\}^{|U|}$.  Now, given two distributions over
functions $\prior, \prior'$, define 
$\Fair(\pin,\pfar,U)$ to be the distribution over {\em labeled}
examples $(y,\ell)$ defined as: with probability $1/2$ choose $y \sim
\pin_U$, $\ell=1$ and with probability $1/2$ choose $y \sim
\pfar_U$, $\ell=0$.  Thus, for a given unlabeled sample $U$, the sets
$\PinC$ and $\PfarC$ define a {\em class} of fair distributions over
labeled examples.  The active testing dimension, roughly,
asks how well this class can be approximated by the class of low-depth
decision trees.  Specifically, let ${\rm DT}_{k}$ denote the class
of decision trees of depth at most $k$.  
The active testing dimension for a given number $u$ of allowed
unlabeled examples is as follows:
\begin{definition}
Given a number $u=poly(n)$ of allowed unlabeled examples, we define
the active testing dimension, $\dactive(u) = \dactive(\calP,D,u)$, as the largest $q \in \nats$
such that 
\[ \sup_{\pin \in \PinC} \sup_{\pfar \in \PfarC} \Pr_{U \sim
D^{u}}({\rm err}^{*}({\rm DT}_{q}, \Fair(\pin,\pfar,U)) < 1/4) \leq
1/4,\] 
where  ${\rm err}^{*}(H,P)$ is the error of the optimal function
in $H$ with respect to data drawn from distribution $P$ over labeled examples.
\end{definition}

\begin{theorem}
\label{thm:activedim}
Active testing of property $\calP$ over distribution $D$ with failure
probability $\frac{1}{8}$ using $u$ 
unlabeled examples requires $\Omega(\dactive(\calP,D,u))$ label queries, and
furthermore can be done with $O(u)$ unlabeled examples and
$O(\dactive(\calP,D,u))$ label queries.
\end{theorem}
\begin{proof}
See Appendix \ref{app:dimension}.
\end{proof}

\noindent
We now use these notions of dimension to prove lower bounds for testing several properties.

\subsection{Application: Dictator functions}
\label{sec:dict}
\newcommand{\bc}{\mathbf{c}}
\newcommand{\bQ}{\mathbf{Q}}

We prove here that active testing of
dictatorships over the uniform distribution requires $\Omega(\log n)$
queries by proving a $\Omega(\log n)$ lower
bound on $\dactive(u)$ for any $u=poly(n)$; in fact, this result holds
even for the specific choice of $\pfar$ as random noise (the uniform
distribution over all functions).
\begin{theorem}
\label{thm:active-dictator}
Active testing of dictatorships under the uniform distribution
requires  $\Omega(\log n)$ queries.  This holds even
for distinguishing dictators from random functions. 
\end{theorem}

\begin{proof}
Define $\pin$ and $\pfar$ to be uniform distributions over the
dictator functions and over all boolean functions, respectively.  In
particular, $\pin$ is the distribution obtained by 
choosing $i \in [n]$ uniformly at random and returning the function $f
: \{0,1\}^n \to \{0,1\}$ defined by $f(x) = x_i$.
Fix $S$ to be a set of $q$ vectors in $\{0,1\}^n$. This set can be
viewed as a $q \times n$ boolean-valued matrix. We write
$c_1(S),\ldots,c_n(S)$  to represent the columns of this matrix.
For any $y \in \{0,1\}^q$,
$$
\pin_S(y) = \frac{|\{i \in [n] : c_i(S) = y\}|}{n} \qquad \mbox{and} \qquad
\pfar_S(y) = 2^{-q}.
$$
By Lemma~\ref{lem:ptlb}, to prove that $\dactive \geq \frac12 \log n$, it
suffices to show that when $q < \frac12 \log n$ and $U$ is a set
of $n^c$ vectors chosen uniformly and independently at random from
$\{0,1\}^n$, then with probability at least $\frac34$, every set $S
\subseteq U$ of size $|S| = q$ and  
every $y \in \{0,1\}^q$ satisfy $\pin_S(y) \le \frac65 2^{-q}$.  (This
is like a stronger version of $\dcoarse$ where $\dist_S(\pin,\pfar)$
is replaced with an $L_\infty$ distance.)

Consider a set $S$ of $q$ vectors chosen uniformly and independently
at random from $\{0,1\}^n$.  For any vector $y \in \{0,1\}^q$, the
expected number of columns of $S$ that are equal to $y$ is 
$n 2^{-q}.$
Since the columns are drawn independently at random, Chernoff bounds
imply that 
$$
\Pr\left[ \pin_S(y) > \tfrac65  2^{-q} \right] \le
e^{-(\frac{1}{5})^2n2^{-q}/3}  < e^{-\frac{1}{75} n 2^{-q}}.
$$
By the union bound, the probability that there exists a vector $y \in
\{0,1\}^q$ such that more than $\frac65 n 2^{-q}$ columns of $S$ are
equal to $y$ is at most $2^q e^{-\frac{1}{75} n 2^{-q}}$.  Furthermore, 
when $U$ is defined as above, we can apply the union bound once again
over all subsets $S \subseteq U$ of size $|S| = q$ to obtain
$
\Pr[ \exists S, y  : \pin_S(y) > \tfrac65  2^{-q} ] < n^{cq} \cdot 2^{q} \cdot 
e^{-\frac{1}{75} n 2^{-q}}.
$
When $q \le \frac 12 \log n$, this probability is bounded above by 
$e^{\frac c 2 \log^2 n + \frac12 \log n - \frac{1}{75} \sqrt{n}}$,
which is less than $\frac14$ when $n$ is large enough, as we wanted to show.
\end{proof}

\subsection{Application: LTFs}
\label{sec:ltf-dim}

The testing dimension also lets us prove the lower bounds in Theorem~\ref{thm:ltf} regarding the query
complexity for testing linear threshold functions. Specifically, those bounds follow directly from the 
following result.

\begin{theorem}
\label{thm:ltf-dimensions}
For linear threshold functions under the standard $n$-dimensional Gaussian distribution,
$\dpassive = \Omega(\sqrt{n/\log(n)})$ and $\dactive = \Omega( (n / \log(n))^{1/3})$.
\end{theorem}

Let us give a brief overview of the strategies used to obtain the $\dpassive$ and $\dactive$ bounds. The complete proofs for 
both results, as well as a simpler proof that $\dcoarse = \Omega( (n/\log n)^{1/3} )$, can be found
in Appendix~\ref{subsec:ltf-dimensions}.

For both results,
we set $\pi$ to be a distribution over LTFs obtained by choosing $w \sim \calN(0,I_{n\times n})$ and outputting 
$f(x) = \sgn(w \cdot x)$. Set $\pi'$ to be the uniform distribution over all functions---i.e., 
for any $x \in \R^n$, the value of $f(x)$ is uniformly drawn from $\{0,1\}$ and is independent of the value of $f$ on other 
inputs. 

To bound $\dpassive$, we bound the total variation distance between the distribution of $Xw / \sqrt{n}$ given $X$,
and a normal $\calN(0,I_{n\times n})$.  If this distance is small, then so must be the distance between
the distribution of $\sgn(Xw)$ and the uniform distribution over label sequences.
In fact, we show this is the case for a broad family of product distributions, characterized by a
condition on the moments of the coordinate projections.

Our strategy for bounding $\dactive$ is very similar to the one we used to prove the lower bound on the query 
complexity for testing dictator functions in the last section. Again, we want to apply Lemma~\ref{lem:ptlb}. 
Specifically, we want to show that when $q \le o((n/\log(n))^{1/3})$ and $U$ is a set of $n^c$ vectors drawn independently 
from the $n$-dimensional standard Gaussian distribution, then with probability at least $\frac34$, every set 
$S \subseteq U$ of size $|S| = q$ and almost all $x \in \R^q$, we have $\pi_S(x) \le \frac65 2^{-q}$.  The 
difference between this case and the lower bound for dictator functions is that we
now rely on strong concentration bounds on the spectrum of random matrices~\cite{Ver11} to obtain
the desired inequality.

\section{Conclusions}
\label{sec:concl}
In this work we develop and analyze a model of property testing that
parallels the active learning model in machine learning, in which
queries are restricted to be selected from a given (polynomially)
large unlabeled sample.  We demonstrate that a number of important
properties for machine learning can be efficiently tested in this
setting with substantially fewer queries than needed to learn.  These
testing algorithms bring together tools from a range of areas
including U-statistics, noise-sensitivity, and self-correction, and
develop characterizations of certain function classes that may be of
independent interest.  We additionally give a combination result
allowing one to build testable properties out of others, as well as
develop notions of intrinsic {\em testing dimension} that characterize
the number of queries needed to test, and which we then use to prove a
number of near-matching lower bounds.   In the context of testing
linear separators, for the active testing model we have an
$\tilde{O}(\sqrt{n})$ upper bound and an $\tilde{\Omega}(n^{1/3})$
lower bound; it would be very exciting if the upper bound could be
improved, but either way it would be interesting to close that gap.
Additionally, testing of linear separators over more general
distributions would be quite interesting.

\bibliographystyle{plain}
\bibliography{active,testing}

\appendix

\section{Comparison of Active Testing and Other Property Testing Models}
\label{app:comparison}
In this section, we compare the active testing model with four existing models of property testing: the standard property testing model as introduced by Rubinfeld and Sudan~\cite{RS96}, the passive testing model first studied by Goldreich, Goldwasser, and Ron~\cite{GGR98}, the tolerant property testing model introduced by Parnas, Ron, and Rubinfeld~\cite{PRR06}, and the distribution-free property testing model of Halevy and Kushilevitz~\cite{HK07}.  

\subsection{Standard and Passive Property Testing}

Fix some sets $X, Y$ and let $\calP$ be some property of functions $f : X \to Y$. Let $D$ be some distribution over $X$.  Recall that the standard model of property testing is defined as follows.

\begin{definition}[Standard Property Tester~\cite{RS96}]
A \emph{$q$-query (standard) $\eps$-tester} for $\calP$ over the distribution $D$ is a randomized algorithm $A$ that queries the value of a function $f$ on $q$ of its inputs and then
\begin{enumerate}
\vspace{-3pt}
\setlength{\itemsep}{1pt}
\setlength{\parskip}{0pt}
\item Accepts with probability at least $\frac23$ when $f \in \calP$, and
\item Rejects with probability at least $\frac23$ when $\dist_D(f,\calP) \ge \eps$.
\end{enumerate}
\end{definition}

The most commonly-studied case is where the distribution $D$ is uniform over the domain of the function.  When that is not the case, note that we can assume that the tester \emph{knows} the distribution $D$.  For the alternate model where the tester does not know $D$, see Section~\ref{app:df}.

\medskip

The \emph{passive} property testing model is similar to the standard property testing model, except that the queries made by the tester in this model are drawn at random from $D$.

\begin{definition}[Passive Property Tester~\cite{GGR98}]
A \emph{$q$-query passive $\eps$-tester} for $\calP$ over the distribution $D$ is a randomized algorithm $A$ that draws $q$ samples independently at random from $D$, queries the value of a function $f$ on each of these $q$ samples, and then
\begin{enumerate}
\vspace{-3pt}
\setlength{\itemsep}{1pt}
\setlength{\parskip}{0pt}
\item Accepts with probability at least $\frac23$ when $f \in \calP$, and
\item Rejects with probability at least $\frac23$ when $\dist_D(f,\calP) \ge \eps$.
\end{enumerate}
\end{definition}

The query complexity of a property under a given testing model is the minimum query complexity of any tester for the property in this model.  We denote the query complexity of properties in the standard, passive, and active testing models with the following notation.

\begin{definition}[Query complexity]
The \emph{query complexity} of $\calP$ over $D$ in the standard property testing model is
$$
Q_{D,\eps}(\calP) := \min\{ q > 0 :  \mbox{there exists a $q$-query $\eps$-tester for $\calP$} \}.
$$
Similarly, the query complexity of $\calP$ over $D$ in the active and passive testing models is
\begin{align*}
Q^a_{D,\eps}(\calP) &:= \min\{ q > 0 :  \mbox{there exists a $q$-query active $\eps$-tester for $\calP$} \} \\
Q^p_{D,\eps}(\calP) &:= \min\{ q > 0 :  \mbox{there exists a $q$-query passive $\eps$-tester for $\calP$} \}.
\end{align*}
\end{definition}

With this notation in place, we can now formally establish the relationship between the standard, active, and passive models of property testing.  

\begin{theorem}
For every property $\calP$, every distribution $D$, and every $\eps > 0$,
\begin{equation}
\label{eqn:hierarchy}
Q_{D,\eps}(\calP) \le Q_{D,\eps}^a(\calP) \le Q_{D,\eps}^p(\calP).
\end{equation}
Furthermore, the three testing models are distinct: there exist properties $\calP$, distributions $D$, and constants $\eps > 0$ such that $Q_{D,\eps}(\calP) < Q_{D,\eps}^a(\calP)$ and there also exist $\calP, D, \eps$ such that $Q_{D,\eps}^a(\calP) < Q_{D,\eps}^p(\calP)$.
\end{theorem}

\begin{proof}
Both inequalities in~\eqref{eqn:hierarchy} are obtained with simple arguments.  For the first inequality, note that we can always simulate an active tester in the standard property testing model by internally sampling\footnote{Note that here we use the fact that a standard property tester knows the underlying distribution $D$ and can therefore generate samples from this distribution ``for free''.} a random subset of the inputs in the domain of the function $f$ and having the active tester select from those inputs.  The second inequality follows from the fact that we can simulate a passive tester in the active testing model by querying the function on the first $Q_{D,\eps}^p(\calP)$ samples drawn at random from $D$.

The distinctness of the three models of property testing is not as immediate, but it follows from the main results in our paper.  Theorem~\ref{thm:active-dictator} shows that testing dictatorship in the active testing model requires $\Omega(\log n)$ queries. The same property can be tested with $O(1/\eps)$ queries in the standard testing model~\cite{BGS98,PRS03}, so this establishes the first strict inequality.  For the second strict inequality, consider the property of being a union of $d$ intervals.  Theorem~\ref{thm:ui-act} shows that we can test this property with $O(1/\eps^4)$ queries in the active testing model but $\Omega(\sqrt{d})$ queries are required to test the same property in the passive model~\cite{KR00}.
\end{proof}

\subsection{Tolerant Testing}

The tolerant property testing model is an extension of the standard model of property testing with one extra requirement: the tester must accept functions with a given property $\calP$ \emph{as well as} functions that are (very) close to $\calP$. Formally, the model is defined as follows.

\begin{definition}[Tolerant Property Tester~\cite{PRR06}]
Fix $0 \le \eps_1 < \eps_2 \le 1$.  A \emph{$q$-query tolerant $(\eps_1,\eps_2)$-tester} for $\calP$ over the distribution $D$ is a randomized algorithm $A$ that queries the value of a funciton $f$ on $q$ of its inputs and then
\begin{enumerate}
\vspace{-3pt}
\setlength{\itemsep}{1pt}
\setlength{\parskip}{0pt}
\item Accepts with probability at least $\frac23$ when $\dist_D(f, \calP) \le \eps_1$, and
\item Rejects with probability at least $\frac23$ when $\dist_D(f, \calP) \ge \eps_2$.
\end{enumerate}
\end{definition}

\newcommand{\tol}{\mathrm{tol}}

\begin{definition}
The \emph{query complexity} of $\calP$ over $D$ in the tolerant property testing model is
$$
Q^{\tol}_{D,\eps_1,\eps_2}(\calP) := \min\{q > 0 : \mbox{there exists a $q$-query tolerant $(\eps_1,\eps_2)$-tester for $\calP$}\}.
$$
\end{definition}


One may ask whether every property that has a query-efficient tolerant tester also has a query-efficient tester in the 
active model. Our lower bound on the query complexity for testing dictator functions in the active model gives a negative answer to this question: there are properties that require significantly more queries to test in the active model than in the tolerant testing model.

\begin{theorem}
\label{thm:active-tolerant}
There exist $\calP$, $D$, and $0 \le \eps_1 < \eps_2 \le 1$ for which $Q^{\tol}_{D,\eps_1,\eps_2}(\calP) < Q^a_{D,\eps_2}(\calP)$.
\end{theorem}

\begin{proof}
Consider the property $\calP$ of being a dictator function and let $D$ be the uniform distribution over the hypercube. Theorem~\ref{thm:active-dictator} shows that $Q^a_{D,\eps_2}(\calP) = \Omega(\log n)$.  By contrast, standard testers for dictator functions~\cite{BGS98,PRS03} are tolerant $(\eps_1,\eps_2)$-testers with query complexity $O(1/(\eps_2-\eps_1)^2)$ so the inequality in the theorem statement holds when $\eps_2 - \eps_1 = \Theta(1)$. 
\end{proof}

We believe that the tolerant and active property testing models are incomparable---i.e., that there exist
properties $\calP$ (along with distributions $D$ and parameters $\eps_1 < \eps_2$) for which the inequality in Theorem~\ref{thm:active-tolerant} is reversed and $Q^{\tol}_{D,\eps_1,\eps_2}(\calP) > Q^a_{D,\eps_2}(\calP)$.  We leave the proof (or disproof) of this assertion as an open problem.

\subsection{Distribution-free testing}
\label{app:df}

In the above property testing models, the tester knows the underlying distribution $D$.  To model the scenario where the tester does not know $D$, Halevy and Kushilevitz~\cite{HK07} introduced the distribution-free testing model. (See also~\cite{HK04,HK05,GS09,DR10}.)  The model is defined formally as follows.

\begin{definition}[Distribution-free Tester~\cite{HK07}]
An \emph{$s$-sample, $q$-query distribution-free $\eps$-tester} for $\calP$ is a randomized algorithm $A$ that draws $s$ independent samples from the (unknown) distribution $D$, queries the value of the (unknown) function $f$ on those $s$ samples and $q-s$ additional inputs of its choosing, and then
\begin{enumerate}
\vspace{-3pt}
\setlength{\itemsep}{1pt}
\setlength{\parskip}{0pt}
\item Accepts with probability at least $\frac23$ when $f \in \calP$, and
\item Rejects with probability at least $\frac23$ when $\dist_D(f,\calP) \ge \eps$.
\end{enumerate}
\end{definition}

\newcommand{\df}{\mathrm{df}}

\begin{definition}
The \emph{query complexity} of the property $\calP$ in the distribution-free model is
$$
Q^{\df}_{\eps}(\calP) := \min\{ q > 0 :  \mbox{for some $0 < s \le q$, there exists an $s$-sample, $q$-query distribution-free $\eps$-tester for $\calP$} \}.
$$
\end{definition}

Superficially, the distribution-free and active testing models appear to be similar: in both models, the tester first samples the underlying distribution $D$ and then queries the value of the function on some inputs.  The challenges in the two models, however, are mostly orthogonal and, as a result, the two models of property testing are incomparable.  This statement is made precise by the following two results.

\begin{theorem}
\label{thm:df}
There exist properties $\calP$ such that for every distribution $D$ and every large enough constant $\eps > 0$, 
$Q^a_{D,\eps}(\calP) < Q^{\df}_{\eps}(\calP)$.
\end{theorem}

\begin{proof}
Fix a large enough $d > 0$ and let $\calP$ be the property consisting of the set of unions of $d$ intervals. Theorem~\ref{thm:ui-act} shows that for every distribution $D$, we have $Q^a_{D,\eps}(\calP) = O(1/\eps^4)$. To complete the
proof of the theorem, we now show that $Q^{\df}_{\eps}(\calP) = \Omega(\sqrt{d})$.

\newcommand{\Fyes}{\mathcal{F}_{\mathrm{yes}}}
\newcommand{\Fno}{\mathcal{F}_{\mathrm{no}}}

Consider the following two distributions on pairs of functions $f : [0,1] \to \{0,1\}$ and distributions $D$ on $[0,1]$.  
For the distribution $\Fyes$, choose a set $S$ of
$d$ points sampled independently and uniformly at random from $[0,1]$.  Define $D$ to be the uniform distribution
over $S$, and let $f : [0,1] \to \{0,1\}$ be a random function defined by choosing $f(x)$ uniformly at random for every
$x \in S$ and setting $f(x) = 0$ for all $x \in [0,1] \setminus S$.  Clearly, every such function $f$ is a union of $d$ intervals.

The distribution $\Fno$ is defined similarly except that in this case we let $S$ be a set of $10d$ points.  We define $D$
to be uniform over $S$ and again define $f : [0,1] \to \{0,1\}$ by choosing $f(x)$ uniformly at random for all $x \in S$ and 
setting $f(x) = 0$ for all remaining points.  In this case, whp the resulting functions are far from unions of $d$ intervals over
$D_2$.

Let $A$ be a distribution-free tester for unions of $d$ intervals. The tester $A$ must accept with high probability when we draw a function $f$ and distribution $D$ from $\Fyes$ and it must reject with high probability when instead we draw a function and distribution from $\Fno$.  Clearly, querying the functions on points that were not drawn from the distribution $D$ will not help $A$ since with probability $1$ it will observe $f(x) = 0$ on those points.  Assume now that $A$ makes $s = o(\sqrt{d})$ draws to the distribution $D$.  By the birthday paradox, with probability $1 - o(1)$, the $s$ samples drawn from the distribution are distinct.  In this case, the distributions on the values of the function on those $s$ inputs are uniformly random so it has no way to distinguish whether the input was drawn from $\Fyes$ or from $\Fno$.  This contradicts the assumption that $A$ is a valid distribution-free tester for unions of $d$ intervals and completes the proof of the lower bound  on $Q^{\df}_{\eps}(\calP)$.
\end{proof}

\begin{theorem}
There exist properties $\calP$, distributions $D$, and parameters $\eps > 0$ such that $Q^{\df}_{\eps}(\calP) < Q^a_{D,\eps}(\calP)$.
\end{theorem}

\begin{proof}
Let $\calP$ be the property of being a dictator function, let $D$ be the uniform distribution over the hypercube, and
let $\eps > 0$ be some constant. Theorem~\ref{thm:active-dictator} shows that $Q^a_{D,\eps}(\calP) = \Omega(\log n)$.  
By contrast, Halevy and Kushilevitz~\cite{HK07} showed that it is possible to test dictator functions in the distribution-free model with a constant number of queries when $\eps$ is constant and so $Q^{\df}_{\eps}(\calP) = O(1)$.
\end{proof}

\section{Proof of a Property Testing Lemma}
\label{app:prelims}
The following lemma is a generalization of a lemma that is widely used
for proving lower bounds in property testing~\cite[Lem.~8.3]{Fis01}.
We use this lemma to prove the lower bounds on the query complexity
for testing dictator functions and testing linear threshold functions.

\begin{lemma}
\label{lem:ptlb}
Let $\pin$ and $\pfar$ be two distributions on functions $X \to \R$.
Fix $U \subseteq X$ to be a set of allowable queries.  
Suppose that for any $S \subseteq U$, $|S|=q$, there is a set $E_S
\subseteq \R^q$ (possibly empty) satisfying $\pin_S(E_S) \le \tfrac15 2^{-q}$ such that
$$
\pin_S(y) < \tfrac65 \pfar_S(y) \mbox{\ \ for every }y \in \R^q \setminus E_S.
$$
Then ${\rm err}^{*}({\rm DT}_q, \Fair(\pin,\pfar,U)) > 1/4$.
\end{lemma}
\begin{proof}
Consider any decision tree $\calA$ of depth $q$.  Each internal node
of the tree consists of a query $y \in U$ and a subset $T \subseteq
\R$ such that its children are labeled by $T$ and $\R \setminus T$,
respectively.  The leaves of the tree are labeled with either
``accept'' or ``reject'', and let $L$ be the set of leaves labeled as
accept.  Each leaf $\ell \in L$ corresponds to a set $S_\ell \subseteq
U^q$ of queries and a subset $T_\ell \subseteq \R^\ell$, where $f : X
\to \R$ leads to the leaf $\ell$ iff $f(S_\ell) \in T_\ell$.  The
probability that $\calA$ (correctly) accepts an input drawn from
$\pin$ is
$$
a_1 = \sum_{\ell \in L} \int_{T_\ell} \pin_{S_\ell}(y) dy.
$$
Similarly, the probability that $\calA$ (incorrectly) accepts an input
drawn from $\pfar$ is 
$$
a_2 = \sum_{\ell \in L} \int_{T_\ell} \pfar_{S_\ell}(y) dy.
$$
The difference between the two rejection probabilities is bounded above by
$$
a_1 - a_2 \le \sum_{\ell \in L} \int_{T_\ell \setminus E_{S_\ell}}
\pin_{S_\ell}(y) - \pfar_{S_\ell}(y) dy  + \sum_{\ell \in L}
\int_{T_\ell \cap E_{S_\ell}} \pin_{S_\ell}(y) dy.
$$
The conditions in the statement of the lemma then imply that
$$
a_1 - a_2 < \sum_{\ell \in L} \int_{T_\ell} \tfrac16 \pin_{S_\ell}(y)
dy + \tfrac56 \sum_{\ell} \int_{E_{S_\ell}} \pin_{S_\ell}(y) dy \le
\tfrac13. 
$$
To complete the proof, we note that $\calA$ errs on an input drawn
from $\Fair(\pin,\pfar,U)$ with probability  
\[
\tfrac12 (1-a_1) + \tfrac12 a_2 = \tfrac12 - \tfrac12(a_1 - a_2) > \tfrac13. \qedhere
\]
\end{proof}

\section{Proofs for Testing Unions of Intervals}
\label{app:ui}
In this section we complete the proofs of the technical results in 
Section~\ref{sec:ui}.

\newtheorem*{propNSInt}{Proposition~\ref{prop:NSInt}}
\begin{propNSInt}[Restated]
Fix $\delta > 0$ and let $f : [0,1] \to \{0,1\}$ be a union of $d$ intervals.  Then 
 $\NS_\delta(f) \le d \delta$.
\end{propNSInt}

\begin{proof}
For any fixed $b \in [0,1]$, the probability that $x < b < y$ when $x \sim U(0,1)$ and $y \sim U(x-\delta, x+\delta)$ is
$$
\Pr_{x,y}[ x < b < y ] = \int_0^\delta  \Pr_{y \sim U(b-t-\delta,b-t+\delta)}[y \ge b] \d t = \int_0^\delta \frac{\delta-t}{2\delta} \d t = \frac{\delta}4.
$$
Similarly, $\Pr_{x,y}[y < b < x] = \frac{\delta}4$.  So the probability that $b$ lies between $x$ and $y$ is at most $\frac{\delta}2$.  

When $f$ is the union of $d$ intervals, $f(x) \neq f(y)$ only if at least one of the boundaries $b_1,\ldots,b_{2d}$ of the intervals of $f$ lies in between $x$ and $y$. So by the union bound, $\Pr[ f(x) \neq f(y)] \le 2d (\delta/2) = d\delta$.  Note that if $b$ is within distance $\delta$ of 0 or 1, the probability is only lower.
\end{proof}

\newtheorem*{lemNSnonInt}{Lemma~\ref{lem:NSnonInt}}
\begin{lemNSnonInt}[Restated]
Fix $\delta = \frac{\eps^2}{32d}$.  Let $f : [0,1] \to \{0,1\}$ be any function with 
noise sensitivity $\NS_\delta(f) \le d\delta(1+\frac{\eps}4)$.  Then $f$ is 
$\eps$-close to a union of $d$ intervals.
\end{lemNSnonInt}

\begin{proof}
The proof proceeds in two steps:  We first show that $f$ is $\frac\eps2$-close to a 
union of $d(1+\frac\eps2)$ intervals, then we show that every union of $d(1 + \frac\eps2)$ 
intervals is $\frac\eps2$-close to a union of $d$ intervals.

Consider the ``smoothed'' function $f_\delta : [0,1] \to [0,1]$ defined by
$$
f_\delta(x) = \E_{y \sim_\delta x} f(y) = 
\frac1{2\delta} \int_{x-\delta}^{x+\delta} f(y) \d y.
$$
The function $f_\delta$ is the convolution of $f$ and the uniform kernel
 $\phi : \R \to [0,1]$ defined by 
$\phi(x) = \frac1{2\delta} \mathbf{1}[ |x| \le \delta ]$.  

Fix $\tau = \frac4\eps \NS_\delta(f)$.  We introduce the function $g^* : [0,1] \to \{0,1,*\}$ by setting
$$
g^*(x) = \begin{cases}
1 & \mbox{when $f_\delta(x) \ge 1 - \tau$,} \\
0 & \mbox{when $f_\delta(x) \le \tau$, and} \\
* & \mbox{otherwise}
\end{cases}
$$
for all $x \in [0,1]$.  Finally, we define $g : [0,1] \to \{0,1\}$ by setting
$g(x) = g^*(y)$ where $y \le x$ is the largest value for which $g(y) \neq *$. (If no
such $y$ exists, we fix $g(x) = 0$.)
\medskip

We first claim that $\distance(f,g) \le \frac\eps2$.  To see this, note that
\begin{align*}
\distance(f,g) &= \Pr_x[ f(x) \neq g(x)]  \\
&\le \Pr_x[ g^*(x) = *] + \Pr_x[ f(x) = 0 \wedge g^*(x) = 1] +
\Pr_x[ f(x) = 1 \wedge g^*(x) = 0] \\
&= \Pr_x[ \tau < f_\delta(x) < 1 - \tau] + \Pr_x[ f(x) = 0 \wedge f_\delta(x) \ge 1 - \tau]
+ \Pr_x[ f(x) = 1 \wedge f_\delta(x) \le \tau].
\end{align*}
We bound the three terms on the RHS individually.  For the first term, 
we observe that $\NSrm_\delta(f,x) = \min \{f_\delta(x), 1 - f_\delta(x)\}$ and
that $\E_x \NSrm_\delta(f,x) = \NS_\delta(f)$.  From these identities and
 Markov's inequality, we have that
$$
\Pr_x[ \tau < f_\delta(x) < 1 - \tau] = \Pr_x[ \NSrm_\delta(f,x) > \tau]
 < \frac{\NS_\delta(f)}\tau = \frac\eps4.
$$
For the second term, let $S \subseteq [0,1]$ denote the set of points $x$ where
$f(x) = 0$ and $f_\delta(x) \ge 1 - \tau$.  Let $\Gamma \subseteq S$ represent
a $\delta$-net of $S$.  Clearly, $|\Gamma| \le \frac1\delta$.  
For $x \in \Gamma$, let $B_x = (x-\delta,x+\delta)$ be a ball
of radius $\delta$ around $x$.  Since $f_\delta(x) \ge 1 - \tau$, the intersection of
$S$ and $B_x$ has mass at most $|S \cap B_x| \le \tau\delta$.  Therefore, the total mass
of $S$ is at most $|S| \le |\Gamma| \tau\delta = \tau$.  By the bounds on the noise
sensitivity of $f$ in the lemma's statement, we therefore have 
$$
\Pr_x[ f(x) = 0 \wedge f_\delta(x) \ge 1 - \tau] \le \tau \le \tfrac\eps8.
$$
Similarly, we obtain the same bound on the third term.  As a result, $\distance(f,g)
\le \frac\eps4 + \frac\eps8 + \frac\eps8 = \frac\eps2$, as we wanted to show.
\medskip

We now want to show that $g$ is a union of $m \le d \delta(1 + \frac\eps2)$ intervals. 
Each left boundary of an interval in $g$ occurs at a point $x \in [0,1]$ where $g^*(x) = *$,
where the maximum $y \le x$ such that $g^*(y) \neq *$ takes the value $g^*(y) = 0$, and
where the minimum $z \ge x$ such that $g^*(z) \neq *$ has the value $g^*(z) = 1$.
In other words, for each left boundary of an interval in $g$, there exists an interval $(y,z)$
such that $f_\delta(y) \le \tau$, $f_\delta(z) \ge 1 - \tau$, and for each $y < x < z$, 
$f_\delta(x) \in (\tau, 1-\tau)$.  Fix any interval $(y,z)$.
Since $f_\delta$ is the convolution of $f$ with a uniform kernel of width $2\delta$,
it is Lipschitz continuous (with Lipschitz constant $\frac1{2\delta}$).  So there exists
$x \in (y, z)$ such that the conditions $f_\delta(x) = \frac12$, $x-y \ge 2\delta(\frac12-\tau)$, and
$z-x \ge 2\delta(\frac12 - \tau)$ all hold.  As a result,
$$
\int_y^z \mathrm{NS}_\delta(f,t) \,\d t = 
\int_y^{x} \mathrm{NS}_\delta(f,t) \,\d t + \int_{x}^z \mathrm{NS}_\delta(f,t) \,\d t \ge
2\delta(\tfrac12 - \tau)^2.
$$
Similarly, for each right boundary of an interval in $g$, we have an interval $(y,z)$ such that
$$
\int_y^z \mathrm{NS}_\delta(f,t) \,\d t \ge
2\delta(\tfrac12 - \tau)^2.
$$
The intervals $(y,z)$ for the left and right boundaries are all disjoints, so
$$
\NS_\delta(f)  \ge \sum_{i=1}^{2m} \int_{y^i}^{z^i} \mathrm{NS}_\delta(f,t) \,\d t \ge 2m \frac\delta2(1-2\tau)^2.
$$
This means that
$$
m \le \frac{d \delta(1 + \eps/4)}{\delta(1 - 2\tau)^2} \le d (1 + \tfrac\eps2)
$$
and $g$ is a union of at most $d(1 + \frac\eps2)$ intervals, as we wanted to show.
\medskip

Finally, we want to show that any function that is the union of
$m \le d(1 + \frac\eps2)$ intervals is $\frac\eps2$-close to a union of $d$ intervals.
Let $\ell_1,\ldots,\ell_m$ represent the lengths of the intervals in $g$.  
Clearly, $\ell_1 + \cdots + \ell_m \le 1$, so there must be a set $S$ of $m -d \le d\eps/2$ 
intervals in $f$ with total length 
$$
\sum_{i \in S} \ell_i \le \frac{m -d}{m} \le \frac{d\eps/2}{d(1 + \frac\eps2)} < \frac\eps2.
$$
Consider the function $h : [0,1] \to \{0,1\}$ obtained by removing the intervals in $S$ 
from $g$ (i.e., by setting $h(x) = 0$ for the values $x \in [b_{2i-1}, b_{2i}]$ for some $i \in S$).  
The function $h$ is a union of $d$ intervals and $\distance(g,h) \le \frac\eps2$.  This completes
the proof, since $\distance(f,h) \le \distance(f,g) + \distance(g,h) \le \eps$.
\end{proof}

\section{Proofs for Testing LTFs}
\label{app:ltf}
\newcommand{\by}{\mathbf{y}}

We complete the proof that LTFs can be tested with $\tilde{O}(\sqrt{n})$ samples in this section.

\subsection{Proof of Lemma~\ref{lem:MORS}}

The proof of Lemma~\ref{lem:MORS} uses the Hermite decomposition of functions. We begin
by introducing this notion and related definitions.

\begin{definition}
The \emph{Hermite polynomials} are a set of polynomials
$h_0(x) = 1, h_1(x) = x, h_2(x) = \tfrac1{\sqrt{2}}(x^2-1),\ldots$
that form a complete orthogonal basis for (square-integrable)
functions $f : \R \to \R$ over the inner product space defined by the
inner product $\left< f, g \right> = \E_{x}[ f(x)g(x) ]$, where the 
expectation is over the standard Gaussian distribution $\calN(0,1)$.
\end{definition}

\begin{definition}
For any $S \in \bbN^n$, define $H_S = \prod_{i=1}^n h_{S_i}(x_i)$.
The \emph{Hermite coefficient} of $f : \R^n \to \R$
corresponding to $S$ is
$
\hat{f}(S) = \left< f, H_S \right> =
\E_{x}[ f(x) H_S(x)]
$
and the \emph{Hermite decomposition} of $f$ is
$
f(x) = \sum_{S \in \bbN^n} \hat{f}(S) H_S(x).
$
The \emph{degree} of the coefficient $\hat{f}(S)$ is $|S| := \sum_{i=1}^n S_i$.
\end{definition}

The connection between linear threshold functions and the Hermite
decomposition of functions is revealed by the following key lemma of
Matulef et al.~\cite{MORS09a}.

\begin{lemma}[Matulef et al.~\cite{MORS09a}]
\label{lem:MORS-orig}
There is an explicit continuous function $W : \R \to \R$ with bounded
derivative $\|W'\|_\infty \le 1$ and peak value $W(0) = \frac2\pi$ such that every
linear threshold function $f : \R^n \to \{-1,1\}$ satisfies
$
\sum_{i=1}^n \hat{f}(e_i)^2 = W( \E_x f ).
$
Moreover, every function $g : \R^n \to \{-1,1\}$ that satisfies
$
\left| \sum_{i=1}^n \hat{g}(e_i)^2 - W( \E_x g) \right| \le 4 \eps^3,
$
is $\eps$-close to being a linear threshold function.
\end{lemma}

In other words, Lemma~\ref{lem:MORS-orig} shows that $\sum_i \hat{f}(e_i)^2$ characterizes
linear threshold functions. 
To obtain Lemma~\ref{lem:MORS}, it suffices to show that this sum is equivalent to
$\E_{x,y}[ f(x) f(y) \left< x,y \right> ]$.  This identity is easily obtained:

\begin{lemma}
\label{lem:sumweights}
For any function $f : \R^n \to \R$, we have
$
\sum_{i=1}^n \hat{f}(e_i)^2 = \E_{x,y}[ f(x) f(y) \left< x, y \right>].
$
\end{lemma}

\begin{proof}
Applying the Hermite decomposition of $f$ and linearity of expectation,
$$
\E_{x,y}[ f(x) f(y) \left<x,y\right>] =
\sum_{i=1}^n \sum_{S,T \in \bbN^n} \hat{f}(S) \hat{f}(T)
\E_x[ H_S(x) x_i] \E_y[ H_T(y) y_i].
$$
By definition, $x_i = h_1(x_i) = H_{e_i}(x)$. The orthonormality of the
Hermite polynomials therefore guarantees that
$
\E_x[ H_S(x) H_{e_i}(x) ] = \mathbf{1}[ S\!=\!e_i ].
$
Similarly, $\E_y[ H_T(y) y_i] = \mathbf{1}[ T\!=\!e_i ]$.
\end{proof}

\subsection{Analysis of {\sc LTF Tester}} 

We now complete the analysis of the {\sc LTF Tester} algorithm.

For a fixed function $f : \R^n \to \R$, define $g : \R^n \times \R^n \to \R$ to be 
$g(x,y) = f(x) f(y) \left< x, y \right>$.
Let $g^* : \R^n \times \R^n \to \R$ be the truncation of $g$ defined by setting
$$
g^*(x,y) = \begin{cases}
f(x) f(y) \left< x, y \right> & \mbox{if } |\left<x,y\right>| \le \sqrt{4n\log(4n/\eps^3)} \\
0 & \mbox{otherwise.}
\end{cases}
$$ 
Our goal is to estimate $\E g$. The following lemma shows that $\E g^*$ provides a good
estimate of this value.

\begin{lemma}
\label{lem:truncation}
Let $g, g^* : \R^n \times \R^n \to \R$ be defined as above. Then $|\E g - \E g^*| \le \tfrac12 \eps^3$.
\end{lemma}

\begin{proof}
For notational clarity, fix $\tau = \sqrt{4n\log(4n/\eps^3)}$. 
By the definition of $g$ and $g^*$ and with the trivial bound $|f(x) f(y) \left< x, y \right>| \le n$ we have
$$
|\E g - \E g^*| = \left| \Pr_{x,y}\big[ \left|\left<x,y\right>\right| > \tau\big] \cdot
\E_{x,y}\Big[ f(x) f(y) \left<x,y\right> \,\big|\, \left|\left<x,y\right>\right| > \tau \Big] \right| 
\le n \cdot \Pr_{x,y}\big[ \left|\left<x,y\right>\right| > \tau\big].
$$
The right-most term can be bounded with a standard Chernoff argument. By Markov's inequality
and the independence of the variables $x_1,\ldots,x_n,y_1,\ldots,y_n$,
$$
\Pr_{x,y}\big[ \left<x,y\right> > \tau\big] = \Pr\big[ e^{t\left<x,y\right>} > e^{t\tau} \big]
\le \frac{\E e^{t\left<x,y\right>}}{e^{t\tau}} = \frac{\prod_{i=1}^n \E e^{t x_i y_i}}{e^{t\tau}}.
$$
The moment generating function of a standard normal
random variable is $\E e^{ty} = e^{t^2/2}$, so
$$
\E_{x_i,y_i}\big[ e^{t x_i y_i} \big] = \E_{x_i} \big[ \E_{y_i} e^{t x_i y_i} \big] = \E_{x_i} e^{(t^2/2) x_i^{\,2}}.
$$
When $x \sim \calN(0,1)$, the random variable $x^2$ has a $\chi^2$ distribution with 1 degree of 
freedom.  The moment generating function of this variable is $\E e^{tx^2} = \sqrt{\frac{1}{1-2t}} = 
\sqrt{1 + \frac{2t}{1 - 2t}}$ for any $t < \frac12$. Hence,
$$
\E_{x_i} e^{(t^2/2) x_i^{\,2}} \le \sqrt{1 + \frac{t^2}{1 - t^2}} \le e^{\frac{t^2}{2(1-t^2)}}
$$
for any $t < 1$. Combining the above results and setting $t = \frac{\tau}{2n}$ yields
$$
\Pr_{x,y}\big[ \left<x,y\right> > \tau\big] \le e^{\frac{n t^2}{2(1-t^2)} - t\tau}
\le e^{-\frac{\tau^2}{4n}} = \tfrac{\eps^3}{4n}.
$$
The same argument shows that $\Pr[ \left<x,y\right> < -\tau ] \le \frac{\eps^3}{4n}$ as well.
\end{proof}

The reason we consider the truncation $g^*$ is that its smaller $\ell_\infty$ norm will enable us to
apply a strong Bernstein-type inequality on the concentration of measure of the U-statistic estimate
of $\E g^*$.

\begin{lemma}[Arcones~\cite{Arc95}] 
\label{lem:arcones2}
For a symmetric function $h : \R^n \times \R^n \to \R$, 
let $\Sigma^2 = \E_x[ \E_y[ h(x,y) ]^2 ] - \E_{x,y}[ h(x,y) ]^2$, let $b = \|h - \E h\|_{\infty}$, and
let $U_m(h)$ be a random variable obtained by drawing $x^1,\ldots,x^m$ independently at random
and setting $U_m(h) = {m \choose 2}^{-1} \sum_{i < j} h(x^{i},x^{j})$.
Then for every $t > 0$,
$$
\Pr[ |U_m(h) - \E h| > t ] \le 4 \exp\left( \frac{ mt^2 }{ 8 \Sigma^2 + 100bt } \right).
$$
\end{lemma}

We are now ready to complete the proof of the upper bound of Theorem~\ref{thm:ltf}.

\begin{theorem}[Upper bound in Theorem~\ref{thm:ltf}, restated]
Linear threshold functions can be tested over the standard $n$-dimensional Gaussian distribution 
with $O(\sqrt{n \log n})$ queries in both the active and passive testing models.
\end{theorem}

\begin{proof}
Consider the \textsc{LTF-Tester} algorithm.  When
the estimates $\tilde\mu$ and $\tilde\nu$ satisfy 
$$
|\tilde\mu - \E f| \le \eps^3   \qquad \mbox{and} \qquad
|\tilde\nu - \E[f(x) f(y) \left<x,y\right>]| \le \eps^3,
$$
Lemmas~\ref{lem:MORS-orig} and~\ref{lem:sumweights} guarantee that the
algorithm correctly distinguishes LTFs from functions that are far from LTFs. To complete the proof,
we must therefore show that the estimates are within the specified error bounds with probability at least
$2/3$.

The values $f(x^1),\ldots,f(x^m)$ are independent $\{-1,1\}$-valued random variables. By Hoeffding's 
inequality,
$$
\Pr[ |\tilde\mu - \E f| \le \eps^3 ] \ge 1 - 2e^{-\eps^6m/2} = 1 - 2e^{-O(\sqrt{n})}.
$$
The estimate $\tilde\nu$ is a U-statistic with kernel $g^*$ as defined above.  This kernel satisfies 
$$
\|g^* - \E g^*\|_\infty \le 2 \|g^*\|_\infty = 2\sqrt{4n \log(4n/\eps^3)}
$$ 
and 
$$
\Sigma^2 \le \E_y\big[ \E_x[ g^*(x,y) ]^2 \big] 
= \E_y \big[ \E_x[ f(x)f(y) \left<x,y\right> \mathbf{1}[\left|\left<x,y\right>\right| \le \tau]  ]^2 \big].
$$
For any two functions $\phi,\psi : \R^n \to \R$, when $\psi$ is $\{0,1\}$-valued the Cauchy-Schwarz inequality implies that $\E_x[ \phi(x) \psi(x) ]^2 \le \E_x[ \phi(x) ] \E_x[ \phi(x) \psi(x)^2 ] = \E_x[\phi(x)] \E_x[ \phi(x) \psi(x) ]$ and so $\E_x[ \phi(x) \psi(x) ]^2 \le \E_x[ \phi(x)]$.  Applying this inequality to the expression
for $\Sigma^2$ gives
$$
\Sigma^2 \le \E_y \big[ \E_x[ f(x)f(y) \left<x,y\right> ]^2 \big] 
=  \E_y \big[ \big( \sum_{i=1}^n f(y) y_i \E_x[ f(x) x_i ]\big)^2 \big] 
= \sum_{i,j} \hat{f}(e_i) \hat{f}(e_j) \E_y [ y_i y_j] = \sum_{i=1}^n \hat{f}(e_i)^2.
$$
By Parseval's identity, we have $\sum_i \hat{f}(e_i)^2 \le \|\hat{f}\|_2^2 = \|f\|_2^2 = 1$. 
Lemmas~\ref{lem:truncation} and~\ref{lem:arcones2} imply that
$$
\Pr[ |\tilde\nu - \E g| \le \eps^3 ] = \Pr[ |\tilde\nu - \E g^*| \le \tfrac12 \eps^3 ]
\ge 1 - 4e^{-\frac{mt^2}{8 + 200\sqrt{n \log(4n/\eps^3)}t}} \ge \tfrac{11}{12}.
$$
The union bound completes the proof of correctness.
\end{proof}

\section{Proofs for Testing Disjoint Unions}
\label{app:disjoint}
\newtheorem*{thmdisjoint}{Theorem~\ref{thm:disjoint}}
\begin{thmdisjoint}[Restated]
Given properties $\calP_1, \ldots, \calP_N$,
if each $\calP_i$ is testable over $D_i$ with $q(\epsilon)$ queries
and $U(\epsilon)$ unlabeled samples, then their disjoint union $\calP$
is testable over the combined distribution $D$ with
$O(q(\epsilon/2)\cdot (\log^3 \frac{1}{\epsilon}))$ queries and
$O(U(\epsilon/2)\cdot (\frac{N}{\epsilon}\log^3\frac{1}{\epsilon}))$
unlabeled samples.
\end{thmdisjoint}

\newcommand{\qd}{q_{\delta}}
\newcommand{\Ud}{U_{\delta}}
\begin{proof}
Let $p = (p_1, \ldots, p_N)$ denote the mixing weights for
distribution $D$; that is, a random draw from $D$ can be viewed as
selecting $i$ from distribution $p$ and then selecting $x$ from $D_i$.
We are given that each $\calP_i$ is testable with failure probability
$1/3$ using using $q(\epsilon)$ queries and $U(\epsilon)$ unlabeled
samples.  By repetition, this implies that each is testable with
failure probability $\delta$ using
$\qd(\epsilon)=O(q(\epsilon)\log(1/\delta))$ queries and 
$\Ud(\epsilon)=O(U(\epsilon)\log(1/\delta))$ unlabeled samples, where
we will set $\delta = \epsilon^2$.
We now test property $\calP$ as follows:

\bigskip

For $\epsilon' = 1/2, 1/4, 1/8, \ldots, \epsilon/2$ do:
\begin{quote}
Repeat $O(\frac{\epsilon'}{\epsilon}\log(1/\epsilon))$ times:
\begin{enumerate}
\setlength{\itemsep}{1pt}
\setlength{\parskip}{0pt}
\item Choose a random $(i,x)$ from $D$.
\item Sample until either $\Ud(\epsilon')$ samples have been drawn
from $D_i$ or $(8N/\epsilon) \Ud(\epsilon')$ samples total have been
drawn from $D$, whichever comes first.
\item In the former case, run the tester for property $\calP_i$ with
parameter $\epsilon'$, making $\qd(\epsilon')$ queries. If the tester
rejects, then reject. 
\end{enumerate}
\end{quote}

If all runs have accepted, then accept.

\bigskip
\noindent
First to analyze the total number of queries and samples, since we can
assume $q(\epsilon) \geq 1/\epsilon$ and $U(\epsilon) \geq
1/\epsilon$, we have $\qd(\epsilon')\epsilon'/\epsilon =
O(\qd(\epsilon/2))$ and $\Ud(\epsilon')\epsilon'/\epsilon =
O(\Ud(\epsilon/2))$ for $\epsilon' \geq \epsilon/2$.  Thus, the total
number of queries made is at most
$$\sum_{\epsilon'} \qd(\epsilon/2)\log(1/\epsilon) = O\left(q(\epsilon/2)
\cdot \log^3 \frac{1}{\epsilon}\right)$$ 
and the total number of unlabeled samples is at most
$$\sum_{\epsilon'} \frac{8N}{\epsilon}\Ud(\epsilon/2)\log(1/\epsilon)
= O\left(U(\epsilon/2) \frac{N}{\epsilon}\log^3\frac{1}{\epsilon}\right).$$
Next, to analyze correctness, if indeed $f \in \calP$ then each call
to a tester rejects with probability at most $\delta$ so the overall
failure probability is at most
$(\delta/\epsilon)\log^2(1/\epsilon) < 1/3$; thus it suffices to
analyze the case that $\distance_D(f,\calP) \geq \epsilon$.  

\medskip
\noindent
If $\distance_D(f,\calP) \geq \epsilon$ then $\sum_{i:p_i \geq
\epsilon/(4N)} p_i \cdot \distance_{D_i}(f_i,\calP_i) \geq 3\epsilon/4$.  
Moreover, for indices $i$ such that $p_i \geq \epsilon/(4N)$, with high
probability Step 2 draws $\Ud(\epsilon')$ samples, so we may
assume for such indices the tester for $\calP_i$ is indeed run in Step 3. 
Let $I = \{i: p_i \geq \epsilon/(4N) \mbox{ and }
\distance_{D_i}(f_i,\calP_i) \geq \epsilon/2\}$.  Thus, we have
$$\sum_{i \in I} p_i \cdot \distance_{D_i}(f_i,\calP_i) \geq \epsilon/4.$$
Let $I_{\epsilon'} = \{i \in I :
\distance_{D_i}(f_i,\calP_i)\in[\epsilon', 2\epsilon']\}$.  Bucketing
the above summation by values $\epsilon'$ in this way implies that for
some value $\epsilon' \in \{\epsilon/2, \epsilon, 2\epsilon,
\ldots, 1/2\}$, we have: $$\sum_{i \in I_{\epsilon'}} p_i \geq
\epsilon/(8\epsilon'\log(1/\epsilon)).$$  This in turn implies that with
probability at least $2/3$, the run of the algorithm for this value of
$\epsilon'$ will find such an $i$ and reject, as desired.
\end{proof}

\section{Proofs for Testing Dimensions}
\label{app:dimension}

\subsection{Passive Testing Dimension (proof of Theorem \ref{thm:passivedim})}

{\bf Lower bound:} 
By design, $\dpassive$ is a lower bound on the number of examples
needed for passive testing.  In particular, if $\dist_S(\pin,\pfar)
\leq 1/4$, and if the target is with probability $1/2$ chosen from
$\pin$ and with probability $1/2$ chosen from $\pfar$, even the Bayes
optimal tester will fail to identify the correct distribution with
probability $\frac{1}{2} \sum_{y \in \{0,1\}^{|S|}} \min(\pin_{S}(y),
\pfar_{S}(y)) = \frac{1}{2}(1 - \dist_S(\pin,\pfar)) \geq 3/8$.   The
definition of $\dpassive$ implies that there exist $\pin \in \PinC$,
$\pfar \in \PfarC$ such that $\Pr_S(\dist_S(\pin,\pfar)\leq 1/4) \geq
3/4$.  Since $\pfar$ has a $1-o(1)$ probability mass on functions that
are $\epsilon$-far from $\calP$, this implies that over random draws
of $S$ and $f$, the overall failure probability of any tester is at least
$(1-o(1))(3/8)(3/4) > 1/4$.  
Thus, at least $\dpassive+1$ random labeled examples are required 
if we wish to guarantee error at most $1/4$.  This in turn implies
$\Omega(\dpassive)$ examples are needed to guarantee error at most $1/3$.

\medskip
\noindent
{\bf Upper bound:} We now argue that $O(\dpassive)$ examples are
\emph{sufficient} for testing as well.  Toward this end, consider the
following natural testing game.  The adversary chooses a function $f$
such that either $f \in \calP$ or $\distance_D(f,\calP)\geq \epsilon$.
The tester picks a function $A$ that maps labeled samples of size $k$
to accept/reject.  That is, $A$ is a deterministic passive testing
algorithm.  The payoff to the tester is the probability that $A$ is
correct when $S$ is chosen iid from $D$ and labeled by $f$.

If $k > \dpassive$ then (by definition of $\dpassive$) we know that for
any distribution $\pin$ over $f \in \calP$ and any distribution
$\pfar$ over $f$ that are $\epsilon$-far from $\calP$, we have 
$\Pr_{S\sim D^k} ( \dist_S(\pin,\pfar) > 1/4) > 1/4$.  
We now need to translate this into a statement about the value of the game.
Note that any mixed strategy of the adversary can be viewed as $\alpha\pin + (1-\alpha)\pfar$ for some distribution $\pin$ over $f \in \calP$, some distribution $\pfar$ over $f$ that are $\epsilon$-far from $\calP$ and some $\alpha \geq 0$. 
The key fact we can use is that against such a mixed strategy,
the Bayes optimal predictor has error exactly
\[\sum_{y} \min( \alpha \pin_{S}(y), (1-\alpha) \pfar_{S}(y))
\leq \max(\alpha,1-\alpha) \sum_{y} \min( \pin_{S}(y), \pfar_{S}(y) ),\]
while
\begin{align*}
\sum_{y} \min( \pin_{S}(y), \pfar_{S}(y) )
& = 1 - (1/2) \sum_{y} | \pin_{S}(y) - \pfar_{S}(y) |
= 1 - \dist_{S}(\pin,\pfar),
\end{align*}
so that the Bayes risk is at most 
$\max(\alpha,1-\alpha) (1 - \dist_{S}(\pin,\pfar))$.
Thus, for any $\alpha \in [7/16,9/16]$, if $\dist_{S}(\pin,\pfar) > 1/4$, 
the Bayes risk is less than $(9/16)(3/4) = 27/64$.
Furthermore, any $\alpha \notin [7/16,9/16]$ has Bayes risk at most $7/16$.
Thus, since $\dist_{S}(\pin,\pfar) > 1/4$ with probability $> 1/4$
(and if $\dist_S(\pin,\pfar)\leq 1/4$ then the error probability of
the Bayes optimal predictor is at most $1/2$), 
for any mixed strategy of the adversary, 
the Bayes optimal predictor has risk less than
$(1/4)(7/16) + (3/4)(1/2) = 31/64$.

Now, applying the minimax theorem we get that for $k = \dpassive+1$,
there exists a mixed strategy $A$ for the tester such that for any
function chosen by the adversary, the probability the tester is
correct is at least $1/2 + \gamma$ for a constant $\gamma >0$ (namely,
$1/64$). 
We can now boost the correctness probability using a constant-factor
larger sample.  Specifically, let $m = c \cdot (\dpassive + 1)$ for
some constant $c$, and consider a 
sample $S$ of size $m$.  The tester simply partitions the sample $S$ into $c$
pieces, runs $A$ separatately on each piece, and then takes majority
vote.  This gives us that $O(\dpassive)$ examples are sufficient for
testing with any desired constant success probability in $(1/2,1)$.  

\subsection{Coarse Active Testing Dimension (proof of Theorem
\ref{thm:coarse})} 

{\bf Lower bound:}
First, we claim that any nonadaptive active testing algorithm that
uses $\leq \dcoarse/c$ 
label requests must use more than $n^c$ unlabeled examples (and thus no 
algorithm can succeed using $o(\dcoarse)$ labels).
To see this, suppose algorithm $A$ draws $n^c$ unlabeled examples.  The number
of subsets of size $\dcoarse/c$ is at most $n^{\dcoarse}/6$ (for
$\dcoarse/c \geq 3$).  So, by definition of $\dcoarse$ and the union bound,
with probability at least $5/6$, \emph{all} such subsets $S$ satisfy the
property that $\dist_S(\pin, \pfar) < 1/4$.  Therefore, for any sequence of
such label requests, the labels observed will not be sufficient to
reliably distinguish $\pin$ from $\pfar$.  Adaptive active testers can
potentially choose their next point to query based on labels observed
so far, but the above immediately implies that even adaptive active
testers cannot use an $o(\log(\dcoarse))$
queries.


\medskip
\noindent
{\bf Upper bound:}
For the upper bound, we modify the argument from the passive 
testing dimension analysis as follows.
We are given that for any distribution $\pin$ over $f \in \calP$ and
any distribution $\pfar$ over $f$ that are $\epsilon$-far from
$\calP$, for $k=\dcoarse+1$, 
we have $\Pr_{S \sim D^k}(\dist_S(\pin,\pfar) >1/4 ) > n^{-k}$.
Thus, we can sample $U \sim D^m$ with $m = \Theta( k \cdot n^{k} )$,
and partition $U$ into subsamples $S_1, S_2, \ldots, S_{cn^k}$ of size
$k$ each.  With high probability, at least one of these subsamples
$S_i$ will have $\dist_S(\pin,\pfar) > 1/4$.  We can thus simply
examine each subsample, identify one such that
$\dist_S(\pin,\pfar)>1/4$, and query the points in that sample.  As in
the proof for the passive bound, this implies that for any strategy
for the adversary in the associated testing game, the best response 
has probability at least $1/2 + \gamma$ of success for some constant
$\gamma>0$.  By the minimax theorem, this implies a testing strategy
with success probability $1/2+\gamma$ which can then be boosted to
$2/3$.  The total number of label requests used in the process is only
$O(\dcoarse)$.

Note, however, that this strategy uses a number of unlabeled examples
$\Omega(n^{\dcoarse+1})$.  Thus, this only implies an active tester
for $\dcoarse = O(1)$.  Nonetheless, combining the upper and lower
bounds yields Theorem \ref{thm:coarse}.

\subsection{Active Testing Dimension (proof of Theorem
\ref{thm:activedim})}

{\bf Lower bound:} for a given sample $U$, we can think of an adaptive
active tester as a decision tree, defined based on which example it
would request the label of next given that the previous requests have
been answered in any given way.  A tester making $k$ queries would
yield a decision tree of depth $k$.  By definition of $\dactive(u)$,
with probability at least $3/4$ (over choice of $U$), any such tester
has error probability at least $(1/4)(1-o(1))$ over the choice of $f$.
Thus, the overall failure probability is at least $(3/4)(1/4)(1-o(1) >
1/8$. 

\medskip
\noindent
{\bf Upper bound:}  We again consider the natural testing game.  We
are given that for any mixed strategy of the adversary with equal
probability mass on functions in $\calP$ and functions $\epsilon$-far
from $\calP$, the best response of the tester has expected payoff at
least $(1/4)(3/4) + (3/4)(1/2) = 9/16$.   This in turn implies that
for any mixed strategy at all, the best response of the tester has
expected payoff at least $33/64$ (if the adversary puts more than
$17/32$ probability mass on either type of function, the tester
can just guess that type with expected payoff at least $17/32$, else
it gets payoff at least $(1 - 1/16)(9/16) > 33/64$).  By the minimax
theorem, this implies existence of a randomized strategy for the
tester with at least this payoff.  We then boost correctness using
$c\cdot u$ samples and $c \cdot \dactive(u)$ queries, running the
tester $c$ times on disjoint samples and taking majority vote.

\subsection{Lower Bounds for Testing LTFs (proof of Theorem~\ref{thm:ltf-dimensions})}
\label{subsec:ltf-dimensions}

We complete the proofs for the lower bounds on the query complexity for testing linear threshold functions in the active and passive models.  
This proof has three parts. First, in Section~\ref{subsec:ltf1}, we introduce some preliminary (technical) results that will be used to prove the lower bounds on the passive and coarse dimensions of testing LTFs. In Section~\ref{subsec:ltf2}, we introduce some more preliminary results regarding random matrices that we will use to bound the active dimension of the class. Finally, in Section~\ref{subsec:ltf3}, we put it all together and complete the proof of Theorem~\ref{thm:ltf-dimensions}.

The high level idea of our proof is we will show that for a random LTF given by weight vector $\mathbf{w} \sim  N(0,  I_{n \times n })$, 
even if we are given the exact value $\mathbf{w} \cdot \mathbf{x}$ for each example $\mathbf{x}$ (rather than just $sgn(\mathbf{w} \cdot \mathbf{x})$), 
we still could not distinguish these values from random Gaussian noise. 
Towards this end, for two distributions $P$, $Q$ over $\mathcal{R}^K$, we use $|| P -Q ||$ to denote the total variation distance between them. 
For instance, given two distributions $\pi$, $\pi'$ over boolean functions, and given a sample $S$, we have $d_S (\pi, \pi') = \| \pi_S - \pi'_S \|$.

\subsubsection{Preliminaries for $\dpassive$ and $\dcoarse$}
\label{subsec:ltf1}

Fix any $K$.
Let the dataset $X = \{x_1, x_2, \cdots, x_{K}\}$ be sampled iid according to a $N(0,I_{n\times n})$ distribution\footnote{In fact, essentially the same argument would work for many other product distributions, including uniform on $\{-1,+1\}^{n}$}.
Let $\mathbf{X} \in \mathcal{R}^{K \times n}$ be the corresponding data matrix. 


Suppose
$\mathbf{w} \sim \calN(0, I_{n \times n})$.
We let
\[\mathbf{z} = \mathbf{X} \mathbf{w},\]
and note that the conditional distribution of $\mathbf{z}$ given $X$
is normal with mean $0$ and ($X$-dependent) covariance matrix, which we denote by $\Sigma$.  
Further applying a threshold function to $\mathbf{z}$ gives $\mathbf{y}$ as the predicted label vector of an LTF.


\begin{lemma}
\label{lemma_2}
For any square non-singular matrix $B$, $\log(det(B)) = Tr(\log(B)),$ where $\log(B)$ is the matrix logarithm of $B$.
\end{lemma}
\begin{proof}
From~\cite{Higham:2008:FM}, we know since every eigenvalue of $A$ corresponds to the eigenvalue of $\exp(A)$, thus
\begin{eqnarray}
\label{eqn_2}
det(\exp(A)) = \exp(Tr(A))
\end{eqnarray}
where $\exp(A)$ is the matrix exponential of $A$. 
Taking logarithm of both sides of~(\ref{eqn_2}), we get 
\begin{eqnarray}
\label{eqn_3}
\log(det(\exp(A))) = Tr(A) 
\end{eqnarray}
Let $B = \exp(A)$ (thus $A = \log(B)$). Then~(\ref{eqn_3}) can rewritten as
$\log(det(B)) = Tr(\log B)$.
\end{proof}

Fixing X, let $P_{\mathbf{z}/\sqrt{n} | X}$ denote the conditional distribution over $\mathbf{z}/\sqrt{n}$ given by choosing $\mathbf{w} \sim N(0, I_{n \times n })$ and letting $\mathbf{z} = \mathbf{X}\mathbf{w}$.

\begin{lemma}
\label{lemma_0}
For sufficiently large $n$, and a value 
$K = \Omega(\sqrt{n / \log(K/\delta)})$,
with probability at least 
$1-\delta$
(over $X$),
\[ \| \P_{(\mathbf{z} / \sqrt{n}) | X} - \calN(0,I) \| \leq 1/4.\]  
\end{lemma}
\begin{proof}
For sufficiently large $n$, for any pair $\mathbf{x}_i$ and $\mathbf{x}_j$,
by Bernstein's inequality, with probability $1-\delta^{\prime}$,
\[
\mathbf{x}_i^T \mathbf{x}_j \in \left[-2\sqrt{n \log\frac{2}{\delta^{\prime}}}, 2\sqrt{ n \log\frac{2}{\delta^{\prime}}} \right]
\]
for $i \neq j$, while concentration inequalities for $\chi^{2}$ random variables \cite{massart:00} imply that with probability $1-\delta^{\prime}$,
\[
\mathbf{x}_i^T \mathbf{x}_j \in \left[n-2\sqrt{n \log\frac{2}{\delta^{\prime}}}, n+ 2\sqrt{ n \log\frac{2}{\delta^{\prime}}} + 2\log\frac{2}{\delta} \right]
\]
for $i = j$.
By the union bound, setting $\delta^{\prime} = \delta / K^2$, the above inclusions hold simultaneously for all $i,j$ with probability at least $1-\delta$.
For the remainder of the proof we suppose this
(probability $1-\delta$)
event occurs.

For $i \neq j$,
\begin{eqnarray*}
Cov(z_i/\sqrt{n}, z_j/\sqrt{n} |X) &=& \frac{\E[z_i z_j |X]}{n} \\
&=& \frac{1}{n}\E\left[(\sum_{l=1}^n w_l \cdot x_{il}) (\sum_{l=1}^n w_l \cdot x_{jl})|X\right]\\
& = & \frac{1}{n}\E\left[\sum_{l,m = 1,1}^{n,n} w_l w_m x_{il} x_{jm} |X \right]\\
& = & \frac{1}{n}\E\left[\sum_{l} w_l^2 x_{il} x_{jl}|X\right] 
  =  \frac{1}{n}\E\left[\sum_l x_{il} x_{jl}|X\right]
\\ &=&  \frac{1}{n} \sum_l x_{il} x_{jl} 
=   \frac{1}{n}\mathbf{x}_i^T \mathbf{x}_j  \in  \left[- 2\sqrt{\frac{\log (2 K^2/\delta)}{n}}, 2\sqrt{\frac{\log (2 K^2/\delta)}{n}}\right] 
\end{eqnarray*}
because $\E[w_l w_m] =0$ (for $l \neq m$) and $\E[w_l^2] =1$.
Similarly, we have 
\[
Var(z_i / \sqrt{n} | X) = \frac{1}{n} \mathbf{x}_i^{T} \mathbf{x}_i \in \left[ 1 - 2\sqrt{\frac{\log\frac{2K^2}{\delta}}{n}}, 1 + 2\sqrt{\frac{\log\frac{2 K^2}{\delta}}{n}} + \frac{2 \log\frac{2 K^2}{\delta}}{n}\right].
\]
Let $\beta = 2\sqrt{\frac{\log(2 K^2/\delta)}{n}} + \frac{2 \log\frac{2 K^2}{\delta}}{n}$.
%
Thus $\Sigma$ is a $K \times K$ matrix, with $\Sigma_{ii} \in [1-\beta,1+\beta]$ for $i = 1, \cdots, K$ and $\Sigma_{ij} \in [-\beta, \beta]$ for all $i \neq j$.

Let $P_1 = \calN(0, \Sigma^{K \times K})$ and $P_2 = \calN(0, I^{K \times K})$.
As the density \[p_1(\mathbf{z})  = \frac{1}{\sqrt{(2 \pi)^K \mbox{det}(\Sigma)}} \exp(-\frac{1}{2} \mathbf{z}^T \Sigma^{-1} \mathbf{z})\]
and the density \[p_2(\mathbf{z}) = \frac{1}{\sqrt{(2\pi)^K}} \exp(-\frac{1}{2}\mathbf{z}^T\mathbf{z})\]
Then $L_1$ distance between the two distributions $P_1$ and $P_2$
\begin{equation*}
 \int |d P_2 - d P_1| \leq 2\sqrt{K(P_1, P_2)} = 2 \sqrt{(1/2) \log \det(\Sigma)},
 \end{equation*}
 where this last equality is by \cite{jason}.
By Lemma~\ref{lemma_2}, $\log(\det(\Sigma)) = Tr(\log (\Sigma))$.  Write $A = \Sigma - I$. By the Taylor series 
\[\log(I + A) = -\sum_{i =1}^\infty \frac{1}{i}(I - (I + A))^i   =  -\sum_{i =1}^\infty \frac{1}{i}(- A)^i\]
Thus,
\begin{eqnarray}
Tr(\log(I + A)) &= &\sum_{i=1}^\infty \frac{1}{i} Tr(( - A)^i).
\label{eqn_6}
\end{eqnarray}
Every entry in $A^i$ can be expressed as a sum of at most $K^{i-1}$ terms, 
each of which can be expressed as a product of exactly $i$ entries from $A$.
Thus, every entry in $A^i$ is in the range $[- K^{i-1} \beta^i, K^{i-1} \beta^i]$.
This means $Tr(A^i) \leq K^i \beta^i$.
Therefore, if $K \beta < 1/2$, since $Tr(A) = 0$, 
the expansion of $Tr(\log(I + A)) \leq \sum_{i=1}^{\infty} K^i \beta^i = O\left(K \sqrt{\frac{\log (K/\delta)}{n}}\right)$.

In particular, for some $K = \Omega(\sqrt{n / \log(K/\delta)})$, 
$Tr(\log(I + A))$ is bounded by the appropriate constant to obtain the stated result.
\end{proof}





\subsubsection{Preliminaries for $\dactive$}
\label{subsec:ltf2}

Given an $n \times m$ matrix $A$ with real entries $\{a_{i,j}\}_{i \in [n], j \in [m]}$, the \emph{adjoint}
(or \emph{transpose} -- the two are equivalent since $A$ contains only real values)
of $A$ is the $m \times n$ matrix $A^*$ whose $(i,j)$-th entry equals $a_{j,i}$.  Let us write
$\lambda_1 \ge \lambda_2 \ge \cdots \ge \lambda_m$ to denote the eigenvalues
of $\sqrt{A^* A}$. These values are the \emph{singular values} of $A$.  
The matrix $A^* A$ is positive semidefinite, so the singular values of $A$ are all non-negative.
We write $\lambda_{\max}(A) = \lambda_1$ and $\lambda_{\min}(A) = \lambda_m$ to represent
its largest and smallest singular values.  Finally, the \emph{induced norm} (or \emph{operator norm})
of $A$ is
$$
\| A \| = \max_{x \in \R^m \setminus \{0\}} \frac{\|Ax \|_2}{\|x\|_2} = \max_{x \in \R^m : \|x\|_2^2 = 1} 
\| Ax \|_2.
$$
For more details on these definitions, see any standard
linear algebra text (e.g.,~\cite{Shi77}). We will also use the following strong concentration bounds
on the singular values of random matrices.

\begin{lemma}[See~{\cite[Cor.\,5.35]{Ver11}}]
\label{lem:randmat}
Let $A$ be an $n \times m$ matrix whose entries are independent standard normal random variables.
Then for any $t > 0$, the singular values of $A$ satisfy
\begin{equation}
\sqrt{n} - \sqrt{m} - t \le \lambda_{\min}(A) \le \lambda_{\max}(A) \le \sqrt{n} + \sqrt{m} + t
\end{equation}
with probability at least $1 - 2e^{-t^2/2}$.
\end{lemma}

The proof of this lemma follows from Talagrand's inequality and Gordon's Theorem for Gaussian
matrices.  See~\cite{Ver11} for the details.  The lemma implies the following corollary
which we will use in the proof of our theorem.

\begin{corollary}
\label{cor:randmat}
Let $A$ be an $n \times m$ matrix whose entries are independent standard normal random variables.
For any $0 < t < \sqrt{n} - \sqrt{m}$, the $m \times m$ matrix $\frac1n A^* A$ satisfies both inequalities
\begin{equation}
\left\| \tfrac1n A^* A - I \right\| \le 3 \frac{\sqrt{m} + t}{\sqrt{n}}
\qquad \mbox{and} \qquad
\det\left( \tfrac1n A^* A \right) \ge e^{-m \left(\frac{(\sqrt{m} + t)^2}{n} + 2\frac{\sqrt{m} + t}{\sqrt{n}}\right)}
\end{equation}
with probability at least $1 - 2e^{-t^2/2}$.
\end{corollary}

\begin{proof}
When there exists $0 < z < 1$ such that $1-z \le \frac1{\sqrt{n}} \lambda_{\max}(A) \le 1 + z$, the
identity $\frac1{\sqrt{n}}\lambda_{\max}(A) = \| \frac1{\sqrt{n}} A \| = \max_{\|x\|_2^2 = 1} \| \frac1{\sqrt{n}} A x \|_2$ implies that
$$
1 - 2z \le (1-z)^2 \le \max_{\|x\|_2^2 = 1} \left\| \tfrac1{\sqrt{n}} A x \right\|_2^2 \le (1 + z)^2 \le 1 + 3z.
$$
These inequalities and the identity
$\|\frac1n A^*A - I\| = \max_{\|x\|_2^2 = 1} \|\frac1{\sqrt{n}} Ax\|_2^2 - 1$ imply that
$
-2z \le \| \tfrac1n A^*A - I \| \le 3z.
$
Fixing $z = \frac{\sqrt{m} + t}{\sqrt{n}}$ and applying Lemma~\ref{lem:randmat} completes the proof
of the first inequality.

Recall that $\lambda_1\le \cdots \le \lambda_m$ are the eigenvalues of $\sqrt{A^* A}$.  Then
$$
\det(\tfrac1n A^* A) = \frac{\det(\sqrt{A^* A})^2}{n} = \frac{(\lambda_1\cdots \lambda_m)^2}{n} \ge \left( \frac{\lambda_1^{\,2}}{n} \right)^m
= \left( \frac{\lambda_{\min}(A)^2}{n} \right)^m.
$$
Lemma~\ref{lem:randmat} and the elementary inequality $1 + x \le e^x$ complete the proof of
the second inequality.
\end{proof}

\subsubsection{Proof of Theorem~\ref{thm:ltf-dimensions}}
\label{subsec:ltf3}

\newtheorem*{thmltfdim}{Theorem~\ref{thm:ltf-dimensions}}
\begin{thmltfdim}[Restated]
For linear threshold functions under the standard Gaussian distribution in 
$\reals^n$,
$\dpassive = \Omega(\sqrt{n/\log(n)})$
and $\dactive = \Omega( ( n / \log(n))^{1/3})$.
\end{thmltfdim}
\begin{proof} 

Let $K$ be as in Lemma~\ref{lemma_0} for $\delta = 1/4$.
Let 
$D = \{(x_1,y_1),\ldots,(x_K,y_K)\}$ denote the sequence of labeled data points under the random LTF based on $\mathbf{w}$.
Furthermore, let $D^{\prime} = \{(x_1,y_1^{\prime}),\ldots,(x_K,y_K^{\prime})\}$ denote the sequence of labeled data points under a target function that assigns an independent random label to each data point.
Also let $\mathbf{z}_i = (1/\sqrt{n}) \mathbf{w}^{T} x_i$, and let $\mathbf{z}^{\prime} \sim N(0,I_{K \times K})$.
Let $E = \{(x_1,\mathbf{z}_1),\ldots,(x_K,\mathbf{z}_K)\}$ and $E^{\prime} = \{(x_1,\mathbf{z}_1^{\prime}),\ldots,(x_K,\mathbf{z}_K^{\prime})\}$.
Note that we can think of $y_i$ and $y_i^{\prime}$ as being functions of $\mathbf{z}_i$ and $\mathbf{z}_i^{\prime}$, respectively.  Thus, letting $X = \{x_1,\ldots,x_K\}$, by Lemma~\ref{lemma_0}, with probability at least $3/4$,
\begin{equation*}
\| \P_{D|X} - \P_{D^{\prime}|X} \| \leq \|\P_{E|X} - \P_{E^{\prime}|X} \| \leq 1/4.
\end{equation*}
This suffices for the claim that $\dpassive = \Omega(K) = \Omega( \sqrt{n / \log(n)} )$.

Next we turn to the lower bound on $\dactive$.
Let us now introduce two distributions $\Dyes$ and $\Dno$ over linear threshold functions and 
functions that (with high probability) are far from linear threshold functions, respectively. We draw
a function $f$ from $\Dyes$ by first drawing a vector $\bw \sim \calN(0,I_{n \times n})$ from the
$n$-dimensional standard normal distribution.  We then define $f : x \mapsto \sgn(\frac1{\sqrt{n}} x \cdot \bw)$.
To draw a function $g$ from $\Dno$, we define $g(x) = \sgn(\by_x)$ where each $\by_x$ variable
is drawn independently from the standard normal distribution $\calN(0,1)$.

Let $\bX \in \R^{n \times q}$ be a random matrix obtained by drawing $q$ vectors from
the $n$-dimensional normal distribution $\calN(0, I_{n \times n})$ and setting these
vectors to be the columns of $\bX$.  
Equivalently, $\bX$ is the random matrix whose entries are independent
standard normal variables.  
When we view $\bX$ as a set of $q$ queries to a function $f \sim \Dyes$
or a function $g \sim \Dno$, we get $f(\bX) = \sgn(\frac1{\sqrt{n}} \bX \bw)$ and 
$g(\bX) = \sgn(\by_{\bX})$. Note
that $\frac1{\sqrt{n}}\bX \bw \sim \calN(0, \frac1n\bX^*\bX)$ and $\by_{\bX} \sim \calN(0, I_{q \times q})$.  To apply
Lemma~\ref{lem:ptlb} it suffices to show that the ratio of the pdfs for both these random variables is
bounded by $\frac65$ for all but $\frac15$ of the probability mass.

The pdf $p : \R^q \to \R$ of a $q$-dimensional random vector from the distribution $\calN_{q \times q}(0,\Sigma)$ is
$$
p(x) = (2\pi)^{-\frac q2} \det(\Sigma)^{-\frac12} e^{-\frac12 x^T \Sigma^{-1} x}.
$$
Therefore, the ratio function $r : \R^q \to \R$ between the pdfs of $\frac1{\sqrt{n}} \bX \bw$ and 
of $\by_{\bX}$ is
$$
r(x) =  \det(\tfrac1n \bX^* \bX)^{-\frac12} e^{\frac12 x^T ((\frac1n \bX^* \bX)^{-1}- I) x}.
$$
Note that 
$$
x^T ((\tfrac1n \bX^* \bX)^{-1}- I) x \le \| (\tfrac1n \bX^* \bX)^{-1} - I \| \|x\|_2^2
= \| \tfrac1n\bX^* \bX - I \| \|x\|_2^2,
$$
so by Lemma~\ref{lem:randmat} with probability at least $1 - 2e^{-t^2/2}$ we have
$$
r(x) \le e^{\frac q2\left(\frac{(\sqrt{q} + t)^2}{n} + 2\frac{\sqrt{q} + t}{\sqrt{n}}\right)
 + 3\frac{\sqrt{q} + t}{\sqrt{n}} \|x\|_2^2}.
$$
By a union bound, for $U \sim \mathcal{N}(0,I_{n\times n})^{u}$, $u \in \nats$ with $u \geq q$,
the above inequality for $r(x)$ is true for all subsets of $U$ of size $q$, 
with probability at least $1 - u^{q} 2 e^{-t^2 / 2}$.
Fix $q = n^{\frac13}/ (50 (\ln(u))^{\frac{1}{3}})$ and $t = 2\sqrt{q \ln(u)}$.  
Then $u^{q} 2 e^{-t^2 / 2} \leq 2 u^{-q}$, which is $< 1/4$ for any sufficiently large $n$.
When $\|x\|_2^2 \le 3q$ then for large $n$,
$r(x) \le e^{74/625} < \frac65$. To complete the proof, it suffices to show that when $x \sim \calN(0,I_{q \times q})$, 
the probability that $\|x\|_2^2 > 3q$ is at most $\frac{1}{5}2^{-q}$. The
random variable $\|x\|_2^2$ has a $\chi^2$ distribution with $q$ degrees of freedom and expected
value $\E \|x\|_2^2 = \sum_{i=1}^q \E x_i^{\,2} = q$.  Standard concentration bounds for $\chi^2$
variables imply that
$$
\Pr_{x \sim \calN(0, I_{q \times q})}[ \|x\|_2^2 > 3q] \le e^{-\frac43 q} < \tfrac15 2^{-q},
$$
as we wanted to show.
Thus, Lemma~\ref{lem:ptlb} implies 
${\rm err}^{*}({\rm DT}_q, \Fair(\pin,\pfar,U)) > 1/4$ holds whenever this $r(x)$ inequality is satisfied for all subsets of $U$ of size $q$;
we have shown this happens with probabiliity greater than $3/4$, so we must have $\dactive \geq q$.
\end{proof}

If we are only interested in bounding $\dcoarse$, the proof can be somewhat simplified.
Specifically, taking $\delta = n^{-K}$ in Lemma~\ref{lemma_0} implies that with probability at least $1-n^{-K}$,
\begin{equation*}
\| \P_{D|X} - \P_{D^{\prime}|X} \| \leq \|\P_{E|X} - \P_{E^{\prime}|X} \| \leq 1/4,
\end{equation*}
which suffices for the claim that 
$\dcoarse = \Omega(K)$, where $K = \Omega( \sqrt{n / K\log(n)})$:
in particular, $\dcoarse = \Omega( (n / \log(n))^{1/3})$.




%



\section{Testing Semi-Supervised Learning Assumptions}
\label{app:margin-proofs}
\label{sec:ssl}
We now consider testing of common assumptions made in semi-supervised
learning \cite{ssl:book06}, where unlabeled data, together with
assumptions about how the target function and data distribution
relate, are used to constrain the search space.   As mentioned in
Section \ref{sec:disjoint},
one such assumption we can test using our generic disjoint-unions
tester is the cluster assumption,  that if data lies in $N$
identifiable clusters, then points in the same cluster
should have the same label.  We can in fact achieve the following
tighter bounds:
\begin{theorem}
We can test the cluster assumption with active testing using
$O(N/\epsilon)$ unlabeled examples and $O(1/\epsilon)$ queries.
\end{theorem}
\begin{proof}
Let $p_{i1}$ and $p_{i0}$ denote the probability mass on positive
examples and negative examples respectively in
cluster $i$, so $p_{i1} + p_{i0}$ is the
total probabilty mass of cluster $i$.  Then $\distance(f,\calP) =
\sum_i \min(p_{i1},p_{i0})$.  Thus, a simple tester is to draw a
random example $x$, draw a random example $y$ from $x$'s cluster, and
check if $f(x)=f(y)$.  Notice that with probability {\em exactly}
$\distance(f,\calP)$, point $x$ is in the minority class of its own
cluster, and conditioned on this event, with probability at least
$1/2$, point $y$ will have a different label.  It thus suffices to
repeat this process $O(1/\epsilon)$ times.  One complication is that
as stated, this process might require a large {\em unlabeled} sample,
especially if $x$ belongs to a cluster $i$ such that $p_{i0}+p_{i1}$
is small, so that many draws are needed to find a point $y$ in $x$'s
cluster.  To achieve the given {\em unlabeled} sample bound, we
initially draw an unlabeled sample of size $O(N/\epsilon)$ and simply
perform the above test on the uniform distribution $U$ over that
sample, with distance parameter $\epsilon/2$.  Standard sample
complexity bounds
\cite{Vapnik:book98} imply that $O(N/\epsilon)$ unlabeled points are
sufficient so that if $\distance_D(f,\calP)\geq
\epsilon$ then with high probability, $\distance_U(f,\calP)\geq \epsilon/2$.
\end{proof}

We now consider the property of a function having a large margin with
respect to the underlying distribution: that is, the distribution $D$
and target $f$ are such that any point in the support of $D|_{f=1}$ is
at distance $\gamma$ or more from any point in the support of
$D|_{f=0}$.  This is a common property assumed in graph-based and
nearest-neighbor-style semi-supervised learning algorithms
\cite{ssl:book06}.  Note that we are not additionally requiring the
target to be a linear separator or have any special functional form.
For scaling, we assume that points lie in the unit ball in $R^d$,
where we view $d$ as constant and $1/\gamma$ as our asymptotic
parameter.
Since we are not
assuming any specific functional form for the target, the number of
labeled examples needed for {\em learning} could be as large as
$\Omega(1/\gamma^d)$ by having a distribution with support over
$\Omega(1/\gamma^d)$ points that are all at distance $\gamma$ from
each other (and therefore can be labeled arbitrarily).
Furthermore, passive testing would require
$\Omega(1/\gamma^{d/2})$ samples as this specific case encodes the
cluster-assumption setting with $N =
\Omega(1/\gamma^d)$ clusters.  We will be able to perform active
testing using only $O(1/\epsilon)$ label requests.

First, one distinction between this and other properties we have
been discussing is that it is a property of the {\em relation} between the
target function $f$ and the distribution $D$; i.e., of the combined
distribution $D_f = (D,f)$ over labeled examples.  As a result, 
the natural notion of {\em distance} to this property is in
terms of the variation distance of $D_f$ to the closest $D_*$
satisfying the property.  As a simple example
illustrating the issue, consider $X = [0,1]$, a target $f$ that is
negative on $[0,1/2)$ and positive on $[1/2,1]$, and a distribution $D$
that is uniform but where the region $[1/2,1/2+\gamma]$ is
downweighted to have total probability mass only $1/2^n$.  Such a
$D_f$ is $1/2^n$-close to the property under variation distance, but
would be nearly $1/2$-far from the property if the 
only operation allowed were to change the function $f$.
A second issue is that we will have to also allow some amount of slack on the
$\gamma$ parameter as well.  Specifically, our tester will distinguish
the case that $D_f$ indeed has margin $\gamma$ from the
case that the $D_f$ is $\epsilon$-far from having margin $\gamma'$
where $\gamma' = \gamma(1-1/c)$ for some constant $c>1$; e.g., think
of $\gamma' = \gamma/2$.  This slack can also be seen to be necessary
(see discussion following the proof of Theorem \ref{thm:margin}).
In particular, we have the following.

\newtheorem*{thmmargin}{Theorem~\ref{thm:margin}}
\begin{thmmargin}[Restated]
For any $\gamma$, $\gamma' = \gamma(1-1/c)$ for constant $c>1$, for
data in the unit ball in $R^d$ for constant $d$, we can
distinguish the case that $D_f$ has margin $\gamma$ from the case that
$D_f$ is $\epsilon$-far from margin $\gamma'$ using Active Testing
with $O(1/(\gamma^{2d}\epsilon^{2}))$ unlabeled examples and
$O(1/\epsilon)$ label requests.
\end{thmmargin}
\newcommand{\witness}{E_{witness}}
\begin{proof}
First, partition the input space $X$ (the unit ball in $R^d$) into
regions $R_1, R_2, \ldots, R_N$ of diameter at most $\gamma/(2c)$.  By
a standard volume 
argument, this can be done using $N = O(1/\gamma^d)$ regions
(absorbing ``$c$'' into the $O()$).  Next, we run the cluster-property
tester on these $N$ regions, with distance parameter $\epsilon/4$.
Clearly, if the cluster-tester rejects, then we can reject as well.
Thus, we may assume below that the total impurity within individual
regions is at most $\epsilon/4$.

Now, consider the following weighted graph $G_\gamma$. We have $N$
vertices, one for each of the $N$ regions.  We have an edge $(i,j)$
between regions $R_i$ and $R_j$ if diam$(R_i \cup R_j) < \gamma$.  We
define the {\em weight} $w(i,j)$ of this edge to be
$\min(D[R_i],D[R_j])$ where $D[R]$ is the probability mass in $R$
under distribution $D$.  Notice that if there is no edge between
region $R_i$ and $R_j$, then by the triangle inequality every point in
$R_i$ must be at distance at least $\gamma'$ from every point in
$R_j$.  Also, note that each vertex has degree $O(c^d) = O(1)$, so the
total weight over all edges is $O(1)$.  Finally, note that while
algorithmically we do not know the edge weights precisely, we can
estimate all edge weights to $\pm \epsilon/(4M)$, where $M=O(N)$ is the
total number of edges, using the unlabeled sample size bounds given in
the Theorem statement.  Let $\tilde{w}(i,j)$ denote the estimated
weight of edge $(i,j)$.

Let $\witness$ be the set of edges $(i,j)$ such that one endpoint is
majority positive and one is majority negative.  Note that if $D_f$
satisfies the $\gamma$-margin property, then every edge in $\witness$
has weight 0. On the other hand, if $D_f$ is $\epsilon$-far from the
$\gamma'$-margin property, then the total weight of edges in
$\witness$ is at least $3\epsilon/4$.  The reason is that otherwise one
could convert $D_f$ to $D'_f$ satisfying the margin condition by
zeroing out the probability mass in the lightest endpoint of every
edge $(i,j) \in \witness$, and then for each vertex, zeroing out the
probability mass of points in the minority label of that vertex.
(Then, renormalize to have total probability 1.)  The first step
moves distance at most $3\epsilon/4$ and the second step moves
distance at most $\epsilon/4$ by our assumption of success of the
cluster-tester.  Finally, if the true total weight of edges in
$\witness$ is at least $3\epsilon/4$ then the sum of their estimated
weights $\tilde{w}(i,j)$ is at least $\epsilon/2$.
This implies we can perform our test as follows.  For $O(1/\epsilon)$
steps, do: 
\begin{enumerate}
\item Choose an edge $(i,j)$ with probability proportional to $\tilde{w}(i,j)$.
\item Request the label for a random $x \in R_i$ and $y\in R_j$.
If the two labels disagree, then reject.
\end{enumerate}
If $D_f$ is $\epsilon$-far from the $\gamma'$-margin property, then
each step has probability $\tilde{w}(\witness)/\tilde{w}(E) = O(\epsilon)$ of
choosing a witness edge, and conditioned on choosing a witness edge
has probability at least $1/2$ of detecting a violation.  Thus,
overall, we can test using $O(1/\epsilon)$ labeled examples and 
$O(1/(\gamma^{2d}\epsilon^2))$ unlabeled examples.
\end{proof}

\medskip
\noindent
{\bf On the necessity of slack in testing the margin assumption:}
Consider an instance space $\X = [0,1]^2$ and two distributions over
labeled examples $D_1$ and $D_2$.  Distribution $D_1$ has probability
mass $1/2^{n+1}$ on positive examples at location $(0, i/2^n)$ and
negative examples at $(\gamma', i/2^n)$ for each $i=1,2,\ldots, 2^n$,
for $\gamma' = \gamma(1 - 1/2^{2n})$.  Notice that $D_1$ is $1/2$-far
from the $\gamma$-margin property because there is a matching between
points in the support of $D_1|_{f=1}$ and points in the support of
$D_1|_{f=0}$ where the matched points have distance less than
$\gamma$.  On the other hand, for each $i=1,2,\ldots,2^n$,
distribution $D_2$ has probability mass $1/2^n$ at either a positive
point $(0, i/2^n)$ {\em or} a negative point $(\gamma', i/2^n)$,
chosen at random, but zero probability mass at the other location.
Distribution $D_2$ satisfies the $\gamma$-margin property, and yet
$D_1$ and $D_2$ cannot be distinguished using a polynomial number of
unlabeled examples.


\end{document}